\documentclass{article}[11pt,twoside]
\usepackage{latexsym}
\usepackage[applemac]{inputenc}
\usepackage{amsfonts}
\usepackage{amssymb}
\usepackage{amstext}
\newcounter{example}
\usepackage{theorem}
\usepackage{pifont}
\usepackage{float}
\usepackage{centernot}
\usepackage{tabularx} 
\usepackage{hyperref}
\usepackage{tabularx} 
\usepackage{booktabs}
\usepackage{varwidth}
\usepackage{mathtools}
\usepackage{fancyhdr}
\usepackage{epigraph}
\newcommand{\Dom}{\mathit{Dom}}

\newcommand{\Coord}{{\cal C}\mathit{oord}}

\newcommand{\qedblack}{\hfill$\blacksquare$}

\newtheorem{claim}{Claim}

\newcommand{\qed}{\hfill$\blacksquare$}

\def\upmapsto{\mathord{\ooalign
    {\hfil\raise .25ex\hbox{%
$\uparrow$}\hfil\crcr
     \hfil\lower .55ex\hbox{$\scriptstyle -$}\hfil }}}
\def\udownmapsto{\mathord{\ooalign
    {\hfil\raise .25ex\hbox{%
$\downarrow$}\hfil\crcr
     \hfil\lower .55ex\hbox{$\scriptstyle -$}\hfil }}}
\newcommand{\inc}{\mathit{inc}}      % Monotonía: Creciente (Increasing)
\newcommand{\inj}{\mathit{inj}}      % Inyectividad (Injective)
\newcommand{\surj}{\mathit{surj}}    % Exhaustividad (Surjective)
\newcommand{\tot}{\mathit{tot}}    % totalidad (total)
\newcommand{\con}{\mathit{con}}    % constancia (contant)
\newcommand{\ntot}{\mathit{non\hbox{-}tot}}    % no totalidad (non-total)
\newcommand{\Inc}{\mathit{Inc}}      
\newcommand{\Dec}{\mathit{Dec}}      
\newcommand{\Inj}{\mathit{Inj}}      
\newcommand{\Surj}{\mathit{Surj}}   
\newcommand{\Tot}{\mathit{Tot}}    
\newcommand{\Con}{\mathit{Con}}

\newcommand{\Ntot}{\mathit{Non\hbox{-}Tot}}

\newtheorem{theorem}{\textbf Theorem}[section]
\newtheorem{definition}[theorem]{\textbf Definition}
\newtheorem{lemma}[theorem]{\textbf Lemma}

\newtheorem{corollary}[theorem]{\bf  Corollary}
\newenvironment{proof}{\par\noindent\textbf{Proof}:\newline}{\par}
\newtheorem{remark}{\bf  Remark}[section]
\newtheorem{example}{\bf  Example}[section]

\newtheorem{observation}{Observation}

\def\<={\trianglelefteq}

\newcommand{\PRE}{\mathsf{PRE}}
\newcommand{\sPRE}{\mathsf{sPRE}}
\newcommand{\PO}{\mathsf{PO}}

\newcommand{\LO}{\mathsf{LO}}
\newcommand{\sLO}{\mathsf{sLO}}
\newcommand{\Or}{\mathsf{O}}
\usepackage{tikz}
\usetikzlibrary{shapes,arrows,automata}

\usepackage{tikz}
\usetikzlibrary{positioning,arrows,calc,shapes,babel} 
\tikzset{
modal/.style={>=stealth',shorten >=1pt,shorten <=1pt,auto,node distance=1.5cm,
semithick},
world/.style={circle,draw,minimum size=0.5cm,fill=gray!15},
point/.style={circle,draw,inner sep=0.5mm,fill=black},
reflexive above/.style={->,loop,looseness=7,in=120,out=60},
reflexive below/.style={->,loop,looseness=7,in=240,out=300},
reflexive left/.style={->,loop,looseness=7,in=150,out=210},
reflexive right/.style={->,loop,looseness=7,in=30,out=330}
}

\begin{document}

\pagestyle{fancy}
\fancyhead{} 
\fancyhead[R]{\thepage} 
\fancyfoot{} 

\title{Definability of Functional Properties in the Basic Modal-Temporal Language over Ordered Frames}

\author{Alfredo Burrieza}
\date{University of Malaga\\
burrieza@uma.es}
\maketitle \thispagestyle{empty}

\begin{abstract}
We study the expressive power of the simplest modal-temporal language, obtained by adding Prior's temporal operators $G$ and $H$ to the basic modal language with $\Box$. This language is the standard bimodal combination of modal and tense logic; under its functional interpretation it is denoted $L_{T\times W}$ in the literature. To analyse its definability across five order types, we consider two semantic readings of the temporal operators: the standard reading ($G,H$), which includes the current instant, and the strict reading ($G^{\ast},H^{\ast}$), which always excludes it.
We examine nine functional properties---totality, non-totality, injectivity, surjectivity, monotonicity, strict monotonicity, antitonicity, strict antitonicity, and constancy---over preorders, strict preorders, partial orders, linear orders, and strict linear orders.
Our analysis reveals two different levels of expressive power. In the original multiflow setting (where multiple functions coexist and the modal operators quantify indiscriminately over their images), the language is quite weak; the two readings of $G,H$ coincide. When we restrict the semantics to minimal functional frames (the $O^{2}$ family), many properties become definable, and the choice of reading becomes crucial: the strict reading can define properties such as injectivity even in reflexive orders.
The same definability patterns appear with indexed languages and with the Uniform Domain (U-Dom) condition on the semantics of $L_{T\times W}$. That three such different ways of controlling functional multiplicity lead to identical definability patterns indicates that the expressive limitations of the original framework come from the uncontrolled multiplicity of functions, not from any weakness of the operators.
Even after controlling functional multiplicity, a set of properties remains undefinable in all non-linear orders, showing that the lack of connectivity is a fundamental obstacle.
\end{abstract}

\section{Introduction}\label{sec:introduction}

We study the definability of functional properties in the simplest modal-temporal language, obtained by adding the temporal operators $G$ and $H$ to the basic modal language with $\Box$. This language is the standard bimodal combination of modal and tense logic (see e.g. \cite{Thomason1984}); in \cite{Burrieza2002conf} it was given a functional interpretation and denoted $L_{T\times W}$.

Throughout the paper we consider two semantic interpretations of the temporal operators $G,H$: the \emph{standard interpretation}, which follows the original semantics (including the current instant when a reflexive loop is present), and the \emph{strict interpretation}, which always excludes the current instant. For clarity, we keep the notation $G,H$ for the standard interpretation, and denote the strict interpretation by $G^{\ast},H^{\ast}$. This notational convention allows us to compare both interpretations without mixing them within the same formula.

In \cite{Burrieza2002conf, Burrieza2003}, the study of functional properties using $L_{T\times W}$ was restricted to strict linear orders and global constraints---such as totality or the Uniform Domain (U-Dom) property. Two questions remained open: whether these properties are definable in isolation, and how their definability varies across a wider range of order structures

The present study addresses these questions by providing a systematic account of the expressive limits of $L_{T\times W}$ across a wide range of order structures, including (strict) linear orders, partial orders, and (strict) preorders. We examine these functional properties independently of the global conditions previously required, such as totality or the Uniform Domain (U-Dom) property, to determine if they are definable in isolation. The functional properties under study are: totality, non-totality, injectivity, surjectivity, monotonicity, strict monotonicity, antitonicity, strict antitonicity, and constancy.

Rather than a mere descriptive map, our approach probes the expressive boundaries of the language to identify how its inherent limitations can be mitigated. To this end, we investigate to what extent the definability range can be expanded through semantic modifications and structural restrictions, without altering the syntax itself.

Our analysis proceeds in two stages. First, we establish the definability limits of $L_{T\times W}$ in the so called \emph{arbitrary multiflow ordered setting.} In this environment, the modal operators function as second-order quantifiers over the images of all functions simultaneously. To investigate this, we maintain a single $L_{T\times W}$ language with its temporal operators $G,H$, but consider two different semantic interpretations: the standard interpretation, for which we keep the notation $G,H$, and the strict interpretation, which we denote by $G^{\ast}, H^{\ast}$ for notational convenience\footnote{The strict interpretation was the only one employed in \cite{Burrieza2002conf, Burrieza2003}.}. This notational convention allows us to compare both interpretations without mixing them within the same formula and demonstrate why they both fail within the multiflow architecture.

To overcome this without altering the syntax, we introduce a semantic refinement: \emph{minimal functional frames} (the $O^2$ family). By restricting each frame to a single accessibility function, we simplify the semantic landscape, allowing the second-order complexity of the multiflow to collapse into a first-order relational reading. Under this restriction, the advantage of the strict interpretation ($G^\ast, H^\ast$) over the standard one becomes evident, as it enables $L_{T\times W}$ to achieve the same definability patterns as more complex indexed languages \cite{Burrieza2009}.

Interestingly, the same definability patterns are obtained when we impose the Uniform Domain (U-Dom) condition on the semantics of $L_{T\times W}$, without abandoning the multiflow architecture. Indexed languages achieve the same result by syntactic means (explicit labels for functions), while U-Dom does it by domain uniformity. Minimal frames, which are a special case of U-Dom, achieve it by structural simplification. That three such different ways of controlling functional multiplicity lead to identical definability patterns shows that the expressive limitations of the original multiflow semantics arise from the indiscriminate quantification over many functions, not from any weakness of the operators of $L_{T\times W}$.

Even after eliminating the interference between functions---whether by structural simplification, domain uniformity, or syntactic labels---a set of properties remains undefinable in all non-linear orders. This reveals that the lack of connectivity is a fundamental obstacle that persists regardless of how the original multiflow semantics is refined.

Methodologically, we establish undefinability for most properties by providing surjective p-morphisms that fail to preserve them. This method is conceptually rooted in the classical preservation theorems of modal logic \cite{GoldblattThomason1974}, which show that surjective p-morphisms preserve validity and can therefore be used to prove undefinability. However, the functional semantics of $L_{T\times W}$ requires extending those results to a non-Kripkean setting, where frames are equipped with both temporal relations and families of functions. Within this framework, we identify a significant exception: in the specific context of \emph{minimal functional frames}, constancy remains invariant under such mappings. Consequently, we employ a distinct semantic equivalence argument to show that $L_{T\times W}$ cannot distinguish between these functional behaviors.

The paper is organised as follows. Section~\ref{sec:functional-frames} introduces the language $L_{T\times W}$, its two temporal interpretations, and the algebraic semantics. Section~\ref{sec:ordered-multiflow} analyses definability in the ordered multiflow setting. Section~\ref{sec:minimal} examines the gains obtained with minimal frames and compares the two temporal interpretations. Section~\ref{sec:conclusions} concludes and outlines directions for future work. Finally, detailed proofs of the results are deferred to the Appendix.

\section{The Language $L_{T\times W}$: Syntax, Semantics and Algebraic Characterisations}\label{sec:functional-frames}

\subsection{Notation and Order-Theoretic Preliminaries}\label{sec:preliminaries}

In what follows, all domains and structures are assumed to be nonempty. 
For a binary relation $R$ on a set $A$, we write 
\[
R(a) \;:=\; \{ a' \in A \mid (a,a') \in R \},\qquad
R^{-1}(a) \;:=\; \{ a' \in A \mid (a',a) \in R \}
\]
for the sets of $R$-images and $R$-preimages of $a$, respectively. 
To handle operators that exclude the evaluation point independently of the properties of $R$ (such as reflexivity), we write
\[
R^{\bullet}(a) := R(a) \setminus \{a\}
\] 
for each $a \in A$. This notation is local: it removes only the point $a$ itself from its own set of successors, leaving all other points unaffected.
\medskip

When $R$ is an order, we adopt standard order-theoretic notation. Let $\le$ be a binary relation on a set $X$; we call it an order when it is at least a \emph{preorder} (reflexive and transitive). If $\le$ is also antisymmetric, it is a \emph{partial order}; if moreover it is total ($x \le y$ or $y \le x$ for all $x, y \in X$), it is a \emph{linear order}. We denote by $<$ the strict part of $\le$, defined by $x < y$ iff $x \le y$ and $x \neq y$. Dually, any irreflexive and transitive relation $<$ is a \emph{strict preorder}, which becomes a \emph{strict partial order} if it is asymmetric, and a \emph{strict linear order} if it satisfies trichotomy.

For any $a \in A$ and $X \subseteq A$, we employ the following operators:
\[
a\!\uparrow \;:=\; \{a' \mid a \leq a'\},\qquad 
a\!\uparrow^\ast \;:=\; \{a' \mid a < a'\}
\]
\[
X\!\uparrow \;:=\; \bigcup_{a \in X} a\!\uparrow,\qquad 
X\!\uparrow^\ast \;:=\; \bigcup_{a \in X} a\!\uparrow^\ast
\]
with $a\!\downarrow$ and $a\!\downarrow^\ast$ defined dually ($a\!\downarrow = \{a' \mid a' \leq a\}$, $a\!\downarrow^\ast = \{a' \mid a' < a\}$).  Note that $a\!\uparrow = a\!\uparrow^\ast \cup\, \{a\}$. By convention, $\varnothing\!\uparrow = \varnothing\!\downarrow = \varnothing\!\uparrow^\ast = \varnothing\!\downarrow^\ast = \varnothing$. 
Furthermore, when $R$ is an order $\leq$, we have $R(a) = a\!\uparrow$ and $R^{\bullet}(a) = a\!\uparrow^\ast$.  In linear orders, these operators coincide with standard interval notation, e.g., $a\!\uparrow = [a, \rightarrow)$ and $a\!\uparrow^\ast = (a, \rightarrow)$. 

\medskip

For a partial function $f\colon A \rightharpoonup B$, we assume the standard definitions for \emph{totality}, \emph{injectivity}, \emph{surjectivity}, and \emph{constancy}. Regarding ordered structures, we refer to increasing and decreasing functions as \emph{monotonicity} and \emph{antitonicity}, respectively, including their \emph{strict} variants.

\medskip
We use function names (e.g., \emph{increasing}) for classes of frames and in the text when referring to the behaviour of individual functions. Property names (e.g., \emph{monotonicity}) are reserved for the logical notion of definability. Formula labels (e.g., $(Inc)$) reflect the function name for consistency.

\subsection{Syntax}

Formulas of $L_{T\times W}$ are generated by the grammar
\[
A ::= \bot \mid p \mid (A \to A) \mid \square A \mid G A \mid H A,
\]
where $p \in \mathcal{V}$. The connectives $\top, \land, \lor, \leftrightarrow, \lozenge, F, P$ are introduced by standard definitions.

Table~\ref{tab:intuitive} summarises the intended readings of the modal and temporal connectives. These readings correspond to the standard interpretation of the operators; the strict interpretation will be introduced semantically in Section~\ref{suc-models}.

\bigskip

\begin{tabularx}{\linewidth}{|c|X|X|}
\hline
\textbf{Conn.} & \textbf{Informal reading} & \textbf{Meaning} \\
\hline
$G$ & Always in the future & $A$ holds at every future moment. \\
\hline
$H$ & Always in the past & $A$ holds at every past moment. \\
\hline
$F$ & Sometime in the future & There exists a future moment where $A$ holds. \\
\hline
$P$ & Sometime in the past & There exists a past moment where $A$ holds. \\
\hline
$\square$ & In all functional images & Every function applicable to the current state has an image satisfying $A$. \\
\hline
$\lozenge$ & In some functional images & There exists a function applicable to the current state whose image satisfies $A$. \\
\hline
\end{tabularx}\label{tab:intuitive}

\subsection{Semantics}\label{subsec-semantics}

\begin{definition}\label{definition1}
A \emph{general functional frame} for $L_{T\times W}$ (or simply a \emph{functional frame})  is a tuple $\Sigma = (W, \mathcal{T}, \mathcal{F})$ where:

\begin{enumerate}
\item $W$ is a nonempty set (of labels).

\item $\mathcal{T} = \{(T_w, R_w) \mid w \in W\}$ is a family such that:
\begin{itemize}
\item[(a)] $T_w$ is a nonempty set for each $w \in W$,
\item[(b)] $T_w \cap T_{w'} = \emptyset$ whenever $w \neq w'$,
\item[(c)] $R_w$ is a binary relation on $T_w$.
\end{itemize}

\item $\mathcal{F}$ is a family of partial functions (\emph{accessibility functions}) such that:
\begin{itemize}
\item[(a)] Each $f_{ww'} \in \mathcal{F}$ is a partial function $f_{ww'}\colon T_w \rightharpoonup T_{w'}$,
\item[(b)] Each $f_{ww'}$ has nonempty domain,
\item[(c)] For each ordered pair $(w, w')$, there is at most one $f_{ww'}$ in $\mathcal{F}$.
\end{itemize}
\end{enumerate}

For each $w \in W$, set $\mathcal{F}_w := \{f_{ww'} \in \mathcal{F} \mid w' \in W\}$, 
so that $\mathcal{F} = \bigcup_{w \in W} \mathcal{F}_w$.
\end{definition}
\noindent
Henceforth, we will simply say ``functional frame'' to refer to this notion, unless a distinction with the more specific formulations in~\cite{Burrieza2002conf} or~\cite{Burrieza2003} is required. The term ``general'' in Definition~\ref{definition1} emphasizes that, at this stage, the relations $R_w$ are arbitrary and not yet restricted to any particular class of orders.

\medskip
\noindent
Note that the definition does not require $\mathcal{F}$ to be closed under composition.

\begin{definition}\label{coordenadas}
Let $\Sigma = (W, \mathcal{T}, \mathcal{F})$ be a functional frame.
\begin{enumerate}
\item Define $\Dom(\mathcal{F}_w) := \bigcup_{f \in \mathcal{F}_w} \Dom(f)$. 

\item The set of \emph{coordinates} of $\Sigma$ is the disjoint union
\[
\Coord_{\Sigma} := \biguplus_{w \in W} T_w.
\]
Convention: an element $t \in T_w$ is identified with the pair $(t,w)$; we write $t_w$ to denote this labelled copy.
\end{enumerate}
\end{definition}

\begin{remark}\label{remark1}
The definition yields the following properties:
\begin{enumerate}
\item[(i)] For any $w \in W$ and any $X \subseteq T_w$, the images $f_{ww'}(X)$ with distinct targets $w'$ are pairwise disjoint. Consequently, the \emph{aggregate image}
\[
\mathcal{F}_w(X) := \bigcup_{f_{ww'}\in \mathcal{F}_w} f_{ww'}(X)
\]
is a disjoint union.

\item[(ii)] The relations $R_w$ are arbitrary; they may be any binary relation. In later sections we will consider frames where every $R_w$ belongs to a fixed class (e.g., linear orders, partial orders, etc.).
\end{enumerate}
\end{remark}

We use the term ``arrow-style'' to refer to inclusions comparing the image of an order segment, such as $a\!\uparrow$, $a\!\uparrow^{*}$, or interval notation, with another segment of the corresponding target order. These conditions take the schematic form
\[
\text{(per-function)}\quad f_{ww'}(X) \subseteq Y,
\qquad\text{or}\qquad
\text{(class-aggregate)}\quad \mathcal{F}_w(X) \subseteq Y,
\]
where $X$ and $Y$ are sets generated by the underlying order. The two forms correspond to two distinct readings: the per-function reading applies the inclusion separately to each accessibility function $f_{ww'}$, while the class-aggregate reading compares instead the combined image $\mathcal{F}_w(X)$ of all functions with source $w$. We will specify which reading is intended whenever this affects the results.

\subsection{Per-function characterisations}

The following theorems establish the algebraic foundation for our definability results by relating functional properties to specific image-set inclusions. While the characterisations for linear orders (Theorem~\ref{fundamentos}) extend the results in \cite{Burrieza2003}, their generalisation to preorders and posets (Theorem~\ref{local}) follows the same logic of order preservation. In both cases, the proofs are straightforward verifications of set-theoretic inclusions under the $\uparrow$ and $\downarrow$ operators and are thus omitted.

\begin{theorem}[Linear orders]\label{fundamentos} 
Let $(A,\leq_A)$ and $(B,\leq_B)$ be nonempty linearly ordered sets, and let 
$f \colon A \rightharpoonup B$ be a partial function with nonempty domain. Then:
\begin{enumerate}
\item $f$ is total iff for all $a\in A$,
\[
   f(a\!\downarrow^\ast) \cup f(a\!\uparrow^\ast)
   \subseteq
   f(\{a\})\!\downarrow^\ast \cup f(\{a\})\!\uparrow.
\]
\item $f$ is injective iff for all $a\in\Dom(f)$,
\[
   f(a\!\downarrow^\ast) \cup f(a\!\uparrow^\ast)
   \subseteq
   f(a)\!\downarrow^\ast \cup f(a)\!\uparrow^\ast.
\]
\item $f$ is surjective iff for all $a\in A$,
\[
   f(\{a\})\!\downarrow^\ast \cup f(\{a\})\!\uparrow^\ast
   \subseteq
   f(a\!\downarrow^\ast) \cup f(a\!\uparrow^\ast).
\]
\item $f$ is increasing (resp. decreasing) iff for all $a\in\Dom(f)$,
\[
   f(a\!\uparrow^\ast)\subseteq f(a)\!\uparrow
   \qquad
   (\text{resp. } f(a\!\uparrow^\ast)\subseteq f(a)\!\downarrow).
\]
\item $f$ is strictly increasing (resp. strictly decreasing) iff for all $a\in\Dom(f)$,
\[
   f(a\!\uparrow^\ast)\subseteq f(a)\!\uparrow^\ast
   \qquad
   (\text{resp. } f(a\!\uparrow^\ast)\subseteq f(a)\!\downarrow^\ast).
\]
\item $f$ is constant iff for all $a\in\Dom(f)$,
\[
   f(a\!\downarrow^\ast) \cup f(a\!\uparrow^\ast)
   \subseteq
   \{f(a)\}.
\]
\end{enumerate}
\end{theorem}

\begin{remark}
The following observations apply to Theorem~\ref{fundamentos}:
\begin{enumerate}
    \item[(i)] Although \emph{non-totality} is among the properties investigated, it is not treated as an independent algebraic notion. Its characterisation is directly tied to the negation of the condition for totality; therefore, no separate algebraic clause is required.
   \item[(ii)] The equivalences remain valid in both non-strict and strict linear orders. On the other hand, the interval notation used in the original formulations (see \cite{Burrieza2003})---e.g., $[f(a), \to)$ and $(f(a), \to)$---adapts automatically to the strict or non-strict nature of the order; here we have opted for a uniform notation with $\uparrow$, $\downarrow$, etc., which already captures the distinction.
\end{enumerate}
\end{remark}

\begin{theorem}[Preorders and posets]\label{local}
Let $(A,\leq_A)$ and $(B,\leq_B)$ be preorders or posets,
and let $f \colon A \rightharpoonup B$ be a partial function with nonempty domain. Then:
\begin{enumerate}
\item $f$ is increasing (resp. decreasing) iff for all $a\in\Dom(f)$,
\[
   f(a\!\uparrow^\ast)\subseteq f(a)\!\uparrow 
   \qquad (\text{resp. } f(a\!\uparrow^\ast)\subseteq f(a)\!\downarrow).
\]
\item $f$ is strictly increasing (resp. strictly decreasing) iff for all $a\in\Dom(f)$,
\[
   f(a\!\uparrow^\ast)\subseteq f(a)\!\uparrow^\ast
   \qquad (\text{resp. } f(a\!\uparrow^\ast)\subseteq f(a)\!\downarrow^\ast).
\]
\end{enumerate}
\end{theorem}

The following theorem generalises the per-function characterisations for totality and surjectivity (items 1 and 3 of Theorem~\ref{fundamentos}) to the multiflow setting. As will be seen, a similar generalisation is not possible for the remaining properties.

\begin{theorem}[Class-aggregate in linear orders]\label{caracterizafunciones}
Let $\Sigma = (W, \mathcal{T}, \mathcal{F})$ be a functional frame in which 
each $(T_w, \leq_w)\in\mathcal{T}$ is linearly ordered. Then:
\begin{enumerate}
  \item $\mathcal{F}$ is a class of total functions iff for all $t_w\in\Coord_\Sigma$,
\[
     \mathcal{F}_w(t_w\!\downarrow^\ast)\cup \mathcal{F}_w(t_w\!\uparrow^\ast)
     \subseteq
     \mathcal{F}_w(\{t_w\})\!\downarrow^\ast \cup \mathcal{F}_w(\{t_w\})\!\uparrow.
  \]
  \item $\mathcal{F}$ is a class of surjective functions iff for all $t_w\in\Coord_\Sigma$,
\[
     \mathcal{F}_w(\{t_w\})\!\downarrow^\ast \cup \mathcal{F}_w(\{t_w\})\!\uparrow^\ast
     \subseteq
     \mathcal{F}_w(t_w\!\downarrow^\ast)\cup \mathcal{F}_w(t_w\!\uparrow^\ast).
  \]
\end{enumerate}
\end{theorem}

\begin{remark}\label{stlo}
The characterisations in Theorem~\ref{local} extend to strict preorders, and those in Theorem~\ref{caracterizafunciones} to strict linear orders, without any change in the statements.
\end{remark}

\begin{example}[Limits of arrow-style aggregate characterisations]\label{fallo}
A natural attempt to extend the aggregate characterisations to other properties would be to replace the function $f$ in the per-function conditions of Theorem~\ref{fundamentos} by the aggregate $\mathcal{F}_w$. For injectivity, the per-function characterisation (Theorem~\ref{fundamentos}(2)) is valid for every $t_w$ in the domain of $f$:
\[
f(t_w\!\downarrow^\ast) \cup f(t_w\!\uparrow^\ast) \subseteq f(t_w)\!\downarrow^\ast \cup f(t_w)\!\uparrow^\ast.
\]
Replacing $f$ by $\mathcal{F}_w$ and requiring $t_w$ to belong to the common domain $\Dom(\mathcal{F}_w)$ yields the candidate condition
\[
\mathcal{F}_w(t_w\!\downarrow^\ast) \cup \mathcal{F}_w(t_w\!\uparrow^\ast) \subseteq \mathcal{F}_w(\{t_w\})\!\downarrow^\ast \cup \mathcal{F}_w(\{t_w\})\!\uparrow^\ast.
\]
The following counterexample shows that this condition does **not** characterise injectivity, and therefore arrow-style aggregate characterisations cannot be extended beyond totality and surjectivity.

\medskip
Let $\Sigma=(W,\mathcal{T},\mathcal{F})$ with $T_w=\{t_w,t'_w\}$, $T_{w'}=\{u_{w'}\}$, $T_{w''}=\{v_{w''}\}$, each strictly ordered: $<_w=\{(t_w,t'_w)\}$, $<_{w'}=<_{w''}=\varnothing$. Define injective functions $f_{ww'}(t_w)=u_{w'}$ and $f_{ww''}(t'_w)=v_{w''}$. 

Since $t'_w \in t_w\!\uparrow^\ast$ and $t_w\!\downarrow^\ast = \varnothing$, we have:
\[
\mathcal{F}_w(t_w\!\uparrow^\ast) \cup \mathcal{F}_w(t_w\!\downarrow^\ast) = \{v_{w''}\}.
\]
However, $\mathcal{F}_w(\{t_w\})=\{u_{w'}\}$. Given that $u_{w'}$ has no strict relations in $T_{w'}$, we obtain:
\[
\mathcal{F}_w(\{t_w\})\!\downarrow^\ast \cup\, \mathcal{F}_w(\{t_w\})\!\uparrow^\ast = \varnothing.
\]
The inclusion $\{v_{w''}\} \subseteq \varnothing$ fails. 

\qedblack
\end{example}

\paragraph{Expressive Limits.}
The scope of arrow-style characterisations is strictly delimited by the connectivity of the underlying order. This approach is intentionally chosen because its reliance on upper and lower intervals matches the reach of the temporal operators in $L_{T\times W}$. While more powerful frameworks exist, such as the equational approach in relation algebras (e.g., \cite{Tarski1987}) where injectivity is defined through composition and converse ($f \circ f^{\smile} \subseteq I$), these require a different modal basis. We discard such alternatives because they operate independently of the order structure; adopting them would necessitate a language with connectives that are not bound by temporal accessibility. Consequently, the arrow-style remains the only framework that is structurally compatible with the inherent locality of our temporal semantics.

\subsection{Functional models and truth}\label{suc-models}

\begin{definition}\label{interpretation}
A \emph{general functional model} for $L_{T\times W}$ (or simply a \emph{functional model}) is a pair $\mathcal{M}=(\Sigma,h)$ where $\Sigma=(W,\mathcal{T},\mathcal{F})$ is a functional frame and $h: L_{T\times W} \longrightarrow 2^{\mathrm{Coord}_{\Sigma}}$ is a \emph{functional interpretation} satisfying:
\begin{align*}
h(\bot) &= \varnothing, \\
h(A\to B) &= (\mathrm{Coord}_{\Sigma}\setminus h(A))\cup h(B),\\
h(GA) &= \{t_w \mid R_w(t_w)\subseteq h(A)\}, \\
h(HA) &= \{t_w \mid R_w^{-1}(t_w)\subseteq h(A)\},\\
h(\square A) &= \{t_w \mid \mathcal{F}_w(\{t_w\})\subseteq h(A)\}.
\end{align*}
\end{definition}

We distinguish two analytical interpretations of the same language:
\begin{itemize}
    \item The \emph{standard interpretation} of $G,H$ is given in Definition~\ref{interpretation}.
    \item The \emph{strict interpretation}, denoted $G^\ast,H^\ast$, evaluates $G$ and $H$ over the irreflexive core $R_w^\bullet$ (see Section~\ref{sec:preliminaries}):
    \[
    h(G^\ast A) = \{t_w \mid R_w^\bullet(t_w) \subseteq h(A)\},\quad
    h(H^\ast A) = \{t_w \mid (R_w^\bullet)^{-1}(t_w) \subseteq h(A)\}.
    \]
\end{itemize}
In irreflexive frames, $R_w = R_w^\bullet$, so both interpretations coincide.

\begin{definition}
Let $(\Sigma,h)$ be a functional model, $t_w \in \mathrm{Coord}_{\Sigma}$, and $A \in L_{T\times W}$.
\begin{itemize}
    \item $A$ is \emph{true at} $t_w$ ($\mathcal{M}, t_w \models A$) if $t_w \in h(A)$.
    \item $A$ is \emph{valid in} $(\Sigma,h)$ if $h(A) = \mathrm{Coord}_{\Sigma}$.
    \item $A$ is \emph{valid in} $\Sigma$ if it is valid in every model on $\Sigma$.
    \item $A$ is \emph{valid} if it is valid in every functional frame.
\end{itemize}
\end{definition}

Having established the general semantic framework, we now turn to definability, relating semantic properties of frames to syntactic characterisations.

\section{Definability in the Arbitrary Ordered Multiflow Setting}\label{sec:ordered-multiflow}
To analyze the expressive limits of $L_{T\times W}$, we first fix some notation. Let $\mathbb{K}$ denote the class of all functional frames. We use $\Or$ for order types and $P$ for functional properties. The expression $\mathbb{K}^{\Or}_{P}$ denotes the class of functional frames in $\mathbb{K}^{\Or}$ where the flows satisfy $\Or$ and the accessibility functions satisfy $P$. 

Consistent with our preliminary notation, we use \emph{property names} (e.g., \emph{monotonicity}, \emph{injectivity}) when referring to the logical notion of definability, while \emph{function names} (e.g., \emph{increasing}, \emph{injective}) describe the behavior of individual functions and identify the corresponding classes. For instance, if $P$ is \emph{monotonicity}, $\mathbb{K}^{\Or}_{\inc}$ is the class of frames whose functions are \emph{increasing}. Throughout the paper, we refer to these order types by their abbreviations: $\mathsf{PRE}$ (preorder), $\mathsf{sPRE}$ (strict preorder), $\mathsf{PO}$ (partial order), $\mathsf{LO}$ (linear order), and $\mathsf{sLO}$ (strict linear order).

The collection of all such structures, where each flow is endowed with a fixed order type $\Or$ without further constraints, constitutes our primary object of study. We call this the \emph{arbitrary ordered multiflow setting}: frames where each flow comes with a fixed order type $\Or$, but without any further restrictions---no limits on the number of flows or functions, no uniformity conditions, and no means to distinguish individual functions. This is the starting point for our analysis; it coincides with the semantics originally introduced in \cite{Burrieza2002conf, Burrieza2003}, except that those works were restricted to strict linear orders.

The functional properties considered in this framework are all properties of a single function:
\begin{center}
    \textit{totality, non-totality, injectivity, surjectivity, monotonicity,}\\
    \textit{strict monotonicity, antitonicity, strict antitonicity, constancy}
\end{center}

We begin by formalising the notion of definability.

\begin{definition}[Definability]\label{def:definability}
Let $L$ be a modal language, and let $\mathbb{K}_2$ be a class of functional frames. For a subclass $\mathbb{K}_1 \subseteq \mathbb{K}_2$:
\begin{itemize}
    \item A set $\Gamma \subseteq L$ of formulas \emph{defines} $\mathbb{K}_1$ \emph{in} $\mathbb{K}_2$ if 
   $$\mathbb{K}_1 = \{ \Sigma \in \mathbb{K}_2 \mid \text{every } A \in \Gamma \text{ is valid in } \Sigma \}.$$\noindent When $\Gamma = \{A\}$, we say that the single formula $A$ \emph{defines} $\mathbb{K}_1$ in $\mathbb{K}_2$.
    
    \item $\mathbb{K}_1$ is \emph{$L$-definable in $\mathbb{K}_2$} if there exists some $\Gamma \subseteq L$ that defines $\mathbb{K}_1$ in $\mathbb{K}_2$.
    
    \item A property $P$ of functions is \emph{$L$-definable in $\mathbb{K}_2$} if the subclass
    $\{ \Sigma \in \mathbb{K}_2 \mid \text{all functions in } \Sigma \text{ satisfy } P \}$
    is $L$-definable in $\mathbb{K}_2$.
\end{itemize}

When the language $L$ is clear from context, we simply say ``definable''.
\end{definition}

\medskip
\noindent\textsc{Convention.}
When we say that a property $P$ is (un)definable in an order type $\Or$, we mean that the subclass $\mathbb{K}^{\Or}_P$ is (un)definable in $\mathbb{K}^{\Or}$ (see Definition~\ref{def:definability}).

This convention is useful because it reflects the fact that definability in a restricted class $\mathbb{K}^{\Or}$ does not automatically transfer to a larger class. For example, totality is definable in $\mathbb{K}^{\sLO}$, but it is not definable in the class of all functional frames $\mathbb{K}$. In contrast, undefinability propagates upwards: if $P$ is undefinable in $\mathbb{K}^{\Or}$, then it is also undefinable in any superclass of $\mathbb{K}^{\Or}$ (in particular, in $\mathbb{K}$). This is why, for instance, the undefinability of totality in $\mathbb{K}^{\sPRE}$ implies its undefinability in all larger classes.

\begin{definition}[Functional bisimulation]\label{bisimulacion}
Let \(\mathcal{M}=(\Sigma,h)\) and \(\mathcal{M}'=(\Sigma',h')\) be functional models.
A nonempty relation \(Z \subseteq \Coord_{\Sigma} \times \Coord_{\Sigma'}\) is a \emph{functional bisimulation} if whenever \(t_w Z t'_{w'}\) the following conditions hold:
\begin{enumerate}
    \item \(t_w \in h(p)\) if and only if \(t'_{w'} \in h'(p)\) for all \(p \in \mathcal{V}\).
    
    \item If \(t_w R_w s_w\), then there exists \(s'_{w'}\) such that \(t'_{w'} R'_{w'} s'_{w'}\) and \(s_w Z s'_{w'}\).
    
    \item If \(t'_{w'} R'_{w'} s'_{w'}\), then there exists \(s_w\) such that \(t_w R_w s_w\) and \(s_w Z s'_{w'}\).
    
    \item If \(s_w R_w t_w\), then there exists \(s'_{w'}\) such that \(s'_{w'} R'_{w'} t'_{w'}\) and \(s_w Z s'_{w'}\).
    
    \item If \(s'_{w'} R'_{w'} t'_{w'}\), then there exists \(s_w\) such that \(s_w R_w t_w\) and \(s_w Z s'_{w'}\).
    
    \item If there exists \(f_{wv} \in \mathcal{F}\) such that \(f_{wv}(t_w) = s_v\), then there exists \(f'_{w'v'} \in \mathcal{F}'\) such that \(f'_{w'v'}(t'_{w'}) = s'_{v'}\) and \(s_v Z s'_{v'}\).
    
    \item If there exists \(f'_{w'v'} \in \mathcal{F}'\) such that \(f'_{w'v'}(t'_{w'}) = s'_{v'}\), then there exists \(f_{wv} \in \mathcal{F}\) such that \(f_{wv}(t_w) = s_v\) and \(s_v Z s'_{v'}\).
\end{enumerate}
\end{definition}
\begin{lemma}[Invariance under functional bisimulation]\label{invarianza}
Let $A \in L_{T\times W}$. If $\mathcal{M}$ and $\mathcal{M}'$ are functional models and $Z$ is a functional bisimulation between them, then for any pair of coordinates $t_w \in \mathrm{Coord}_{\Sigma}$ and $t'_{w'} \in \mathrm{Coord}_{\Sigma'}$ such that $t_w Z t'_{w'}$, it holds that:
\[
\mathcal{M}, t_w \models A \quad \text{if and only if } \quad \mathcal{M}', t'_{w'} \models A.
\]
In particular, bisimilar models satisfy the same formulas at related coordinates (cf. the standard invariance result for modal logic \cite{vanBenthem1984}).
\end{lemma}
\begin{proof}
By structural induction on $A$. \qed
\end{proof}
\medskip

Definability results will be established by providing explicit formulas for the intended properties. Conversely, with a single later exception, undefinability will be witnessed via surjective p-morphisms which, by inducing functional bisimulations (Lemma~\ref{invarianza}), ensure the preservation of validity.

\medskip
\noindent\textsc{Methodological Note.}
Undefinability will be witnessed via surjective p-morphisms $\varphi: \Sigma \to \Sigma'$ between finite frames, where $\Sigma$ satisfies a property $P$ and $\Sigma'$ does not. Since these p-morphisms induce functional bisimulations (Lemma~\ref{invarianza}), they ensure that no formula in $L_{T\times W}$ can define $P$.

Consequently, it suffices to present the relevant pairs of frames. For the sake of brevity, we employ \emph{Convention (RC)} (Reflexive Closure) to adapt these counterexamples to reflexive classes. Detailed graphical representations and specific reading conventions are provided in the Appendix (see~\ref{appendix:figures}).

\medskip
To facilitate the translation of functional properties into modal formulas, we begin with a simple set-theoretic observation.
\begin{observation}\label{obs:singletons}
Taking Remark~\ref{remark1}(i) into account, for every $X \subseteq T_w$ we have:
\[
  \mathcal{F}_w(X) = \bigcup_{t_w \in X} \mathcal{F}_w(\{t_w\}).
\]
\end{observation}

\begin{lemma}[Image-modal correspondence]\label{lem:image-modal-unified}
Let $\Sigma=(W,\mathcal T,\mathcal{F})$ be a functional frame, fix $T_w\in\mathcal T$ and $X\subseteq T_w$.
For every interpretation $h$ and every formula $A$, the following equivalences hold:

\begin{itemize}
\item For $\mathsf{sPRE} $ and $\mathsf{sLO}$:
\begin{enumerate}
  \item $X\subseteq h(\square A)$ if and only if $\mathcal{F}_w(X)\subseteq h(A)$.
  \item $\mathcal{F}_w(X)\!\uparrow^{\ast}\subseteq h(A)$ if and only if $\mathcal{F}_w(X)\subseteq h(GA)$.
  \item $\mathcal{F}_w(X)\!\downarrow^{\ast}\subseteq h(A)$ if and only if $\mathcal{F}_w(X)\subseteq h(HA)$.
\end{enumerate}
\item For $\mathsf{PRE}, \mathsf{PO}$, and $\mathsf{LO}$: The same equivalences hold, but in (2) and (3) replace $\uparrow^\ast$ by $\uparrow$ and $\downarrow^{\ast}$ by $\downarrow$.
\end{itemize}
\end{lemma}

\begin{proof}
By Remark~\ref{remark1}, the union $\mathcal{F}_w(X)=\bigcup_{f\in\mathcal{F}_w} f(X)$ is disjoint; by Observation~\ref{obs:singletons}, we may equivalently write $\mathcal{F}_w(X)=\bigcup_{t_w\in X} \mathcal{F}_w(\{t_w\})$.

\medskip\noindent
\textit{Proof of (1):} The following chain of equivalences holds:
\[
\begin{aligned}
X\subseteq h(\square A)
& \ \text{iff}\ \Big(\bigcup_{t_w\in X} \{t_w\}\Big)\subseteq h(\square A)\\
& \ \text{iff}\ \Big(\bigcup_{t_w\in X} \mathcal{F}_w(\{t_w\})\Big)\subseteq h(A)\\
& \ \text{iff}\ \mathcal{F}_w(X)\subseteq h(A).
\end{aligned}
\]

\medskip\noindent
\textit{Proof of (2) (the proof of (3) is analogous):} By Observation~\ref{obs:singletons}, 
\[
\mathcal{F}_w(X)\!\uparrow^\ast = \bigcup_{t_w\in X} \mathcal{F}_w(\{t_w\})\!\uparrow^\ast.
\] 
Then:
\[
\begin{aligned}
\mathcal{F}_w(X)\!\uparrow^\ast\subseteq h(A)
& \ \text{iff}\ \Big(\bigcup_{t_w\in X} \mathcal{F}_w(\{t_w\})\!\uparrow^\ast\Big)\subseteq h(A)\\
& \ \text{iff}\ \Big(\bigcup_{t_w\in X} \mathcal{F}_w(\{t_w\})\Big)\subseteq h(GA)\\
& \ \text{iff}\ \mathcal{F}_w(X)\subseteq h(GA).
\end{aligned}
\]

\medskip\noindent
For $\mathsf{PRE}, \mathsf{PO}$, and $\mathsf{LO}$, the same proofs apply with $\uparrow$ and $\downarrow$ in place of $\uparrow^\ast$ and $\downarrow^\ast$.
\end{proof}

\begin{theorem}[On Definability of Totality and Surjectivity]\label{tot-surj}
Each of totality and surjectivity is definable only in $\mathsf{LO}$ and $\mathsf{sLO}$.
\end{theorem}

\begin{proof}
We divide the proof into positive cases (definability) and negative cases (undefinability). 

\medskip
\noindent\textsc{i. Positive cases: Definability in linear and strict linear orders.}  

\noindent
\emph{Totality.} The formula
\[
(\Tot): \square(Hp \land p \land Gp) \to (H\square p \land G\square p)
\]
defines $\mathbb{K}^{\sLO^2}_{\tot}$ in $\mathbb{K}^{\sLO^2}$. In $\LO^2$, reflexivity makes the conjunct $p$ redundant, so the simplified formula
\[
(\Tot)^\circ: \square(Hp \land Gp) \to (H\square p \land G\square p)
\]
defines $\mathbb{K}^{\LO^2}_{\tot}$ in $\mathbb{K}^{\LO^2}$ (the proof for totality in strict linear orders can be found in \cite{Burrieza2003}.)

\noindent
\emph{Surjectivity.} We show that the following formula defines $\mathbb{K}^\mathsf{sLO}_{\surj}$ in $\mathbb{K}^\mathsf{sLO}$ (it is the same formula for $\LO$):
\[
(\mathit{Surj}) \quad (H\square p \land G\square p) \to \square(Hp \land Gp).
\]
$(\Rightarrow$) Suppose that  $\Sigma\in \mathbb{K}^\mathsf{sLO}_\surj$. Let $(\Sigma, h)$  be a model and let $t_w\in \Coord_{\Sigma}$. Then:
\[
\begin{aligned}
t_w \in h(H\square p \land G\square p)
&\ \text{iff}\ t_w \in h(H\square p)\cap h(G\square p)\\
&\ \text{iff}\ t_w\!\downarrow^\ast \subseteq h(\square p)\ \text{and}\ t_w\!\uparrow^\ast\subseteq h(\square p) %\quad\text{(Remark~\ref{linear})}
\\
&\ \text{iff}\  (t_w\!\downarrow^\ast\!\cup\,  t_w\!\uparrow^\ast) \subseteq h(\square p)\\
&\ \text{iff}\ \mathcal{F}_w(t_w\!\downarrow^\ast)\cup\mathcal{F}_w( t_w\!\uparrow^\ast)\subseteq h(p)
\quad\text{(Lemma~\ref{lem:image-modal-unified}(1))}\\
&\ \text{then}\ \mathcal{F}_w(\{t_w\})\!\downarrow^{\ast}\!\cup\,\mathcal{F}_w(\{t_w\})\!\uparrow^{\ast}\subseteq h(p) \quad\text{(Theorem~\ref{caracterizafunciones}(2))}\\
&\ \text{iff}\ \mathcal{F}_w(\{t_w\})\!\downarrow^{\ast}\subseteq h(p)\ \text{and}\ \mathcal{F}_w(\{t_w\})\!\uparrow^{\ast}\subseteq h(p)\\
&\ \text{iff}\ \mathcal{F}_w(\{t_w\})\subseteq h(Hp)\cap h(Gp) \quad\text{(Lemma~\ref{lem:image-modal-unified}(2),(3))}\\
&\ \text{iff}\ \mathcal{F}_w(\{t_w\})\subseteq h(Hp \land Gp) \\
&\ \text{iff}\ \{t_w\} \subseteq h(\square(Hp \land Gp)) \quad \text{(Lemma~\ref{lem:image-modal-unified}(1))}\\
&\ \text{iff}\ t_w \in h(\square(Hp \land Gp)).
\end{aligned}
\]
Hence $(\mathit{Surj})$ is valid in $\Sigma$.

\smallskip\noindent
($\Leftarrow$) By contraposition: assume $\Sigma \notin \mathbb{K}^{\mathsf{sLO}}_\surj$. By Theorem~\ref{caracterizafunciones}(2), there exists $t_w$ such that:
\[
  \mathcal{F}_w(\{t_w\})\!\downarrow^\ast \cup \ \mathcal{F}_w(\{t_w\})\!\uparrow^\ast \nsubseteq \mathcal{F}_w(t_w\!\downarrow^\ast) \cup\, \mathcal{F}_w (t_w\!\uparrow^\ast)
  \tag{$\dag$}\]
Define a model $(\Sigma, h)$, where $h(p) = \mathcal{F}_w(t_w\!\downarrow^\ast) \cup \,  \mathcal{F}_w(t_w\! \uparrow^\ast)$. By Lemma~\ref{lem:image-modal-unified}(1), we have $t_w \!\downarrow^\ast\!\cup \, t_w\! \uparrow^\ast\, \subseteq h(\square p)$, which implies $t_w \in h(H\square p \land G\square p)$. However, by $(\dag)$, the condition $\mathcal{F}_w(\{t_w\})\!\downarrow^\ast\!\cup \, \mathcal{F}_w(\{t_w\})\!\uparrow^\ast \subseteq h(p)$ fails. Thus $t_w \notin h(\square(Hp \land Gp))$, and $(\Surj)$ is invalid in $\Sigma$.

$\mathsf{LO}$\emph{ case.} The proof for $\mathsf{LO}$ follows the same structure as for $\mathsf{sLO}$. The only difference arises from the semantics of $G$ and $H$: in reflexive orders, $t_w \in h(H\square p \land G\square p)$ already guarantees $\mathcal{F}_w(\{t_w\}) \subseteq h(p)$. This allows us to augment the strict inclusions $\mathcal{F}_w(\{t_w\})\!\downarrow^{\ast} \subseteq h(p)$ and $\mathcal{F}_w(\{t_w\})\!\uparrow^{\ast} \subseteq h(p)$ with the point itself, yielding $\mathcal{F}_w(\{t_w\}) \subseteq h(Hp) \cap h(Gp)$ in the reflexive sense. The remainder of the proof then proceeds identically to the $\mathsf{sLO}$ case.
\qed
\end{proof}

\medskip\noindent
\textsc{ii. Negative cases: Undefinability in preorders, strict preorders, and posets.}

Totality and surjectivity are undefinable in $\sPRE$ as witnessed by the following counterexamples in the Appendix:

\medskip\noindent
\emph{Totality:} Figure~\ref{fig:total}.

\medskip\noindent
\emph{Surjectivity:} Figure~\ref{fig:surj}.

\medskip\noindent
By Convention~(RC), these configurations also establish undefinability in $\PRE$ and $\PO$. \qed

\medskip
We now proceed to examine other functional properties whose definability fails uniformly across all the order types under consideration. 

\begin{theorem}[Undefinability beyond totality and surjectivity]\label{thm:tot-surj}
Each of non-totality, injectivity, constancy, monotonicity, antitonicity, strict monotonicity, and strict antitonicity is not definable in any of the order types considered ($\PRE$, $\sPRE$, $\PO$, $\LO$, $\sLO$).
\end{theorem}

\begin{proof}
Undefinability in $\sPRE$ and $\sLO$ is established via the surjective frame homomorphisms illustrated in the following figures:

\medskip\noindent
\emph{Non-totality, injectivity, and strict monotonicity:} Figure~\ref{fig:inj_SPO-SOL}.

\medskip\noindent
\emph{Monotonicity  and constancy:} Figure~\ref{fig:inc}.

\medskip\noindent
\emph{Antitonicity and strict antitonicity:} Figure~\ref{fig:dec}.

\medskip\noindent
By Convention~(RC), the undefinability of these properties (including constancy and all forms of monotonicity/antitonicity) extends to $\PRE, \PO$, and $\LO$. \qed
\end{proof}

\medskip
The preceding results yield the following complete classification of the expressive capacity of $L_{T\times W}$:

\begin{theorem}[Main Classification Theorem]\label{main}
Let $P$ be one of the nine functional properties considered in this paper (totality, non-totality, injectivity, surjectivity, monotonicity, strict monotonicity, antitonicity, strict antitonicity, constancy) and $\Or \in \{\PRE, \sPRE, \PO, \LO, \sLO\}$. The property $P$ is definable in $\Or$ if and only if:
\begin{enumerate}
  \item $P \in \{$totality, surjectivity$\}$, and
  \item $\Or \in \{\LO, \sLO\}$.
\end{enumerate}
\end{theorem}

Table~\ref{tab:definability-summary} summarizes these results. A value of \texttt{Yes} indicates that the property is definable in the corresponding order type, while \texttt{No} indicates it is not.

\begin{table}[ht]
\centering
\small
\caption{Summary of definability in the $L_{T\times W}$ multiflow setting.}
\label{tab:definability-summary}
\renewcommand{\arraystretch}{1.2}
\setlength{\tabcolsep}{0pt}
\begin{tabular*}{0.9\textwidth}{@{\extracolsep{\fill}} l ccccc @{}}
\toprule
\textbf{Functional property} & \textbf{PRE} & \textbf{sPRE} & \textbf{PO} & \textbf{LO} & \textbf{sLO} \\
\midrule
Totality             & No & No & No & Yes & Yes \\
Non-totality         & No & No & No & No  & No  \\
Injectivity          & No & No & No & No  & No  \\
Surjectivity         & No & No & No & Yes & Yes \\
Monotonicity         & No & No & No & No  & No  \\
Strict monotonicity  & No & No & No & No  & No  \\
Antitonicity         & No & No & No & No  & No  \\
Strict antitonicity  & No & No & No & No  & No  \\
Constancy            & No & No & No & No  & No  \\
\bottomrule
\end{tabular*}
\end{table}

\subsection{Strict vs. Standard interpretation: Expressive divergence}

\begin{theorem}\label{thm:non-interdef}
In the purely temporal setting, the standard and strict interpretations of $L_{T\times W}$ are not expressively equivalent. Specifically, there is no formula that remains invariant when switching between the standard ($G,H$) and the strict ($G^\ast, H^\ast$) interpretations across all models.
\end{theorem}
\begin{proof}
\noindent\emph{(i) The standard interpretation cannot be captured by the strict interpretation.}
Let $\mathcal{M}$ be a model with a single point $t$ such that $tRt$, and let $\mathcal{N}$ be a model with a single point $t'$ such that $\neg(t'Rt')$. In both models, every atom is false. The irreflexive core is empty ($R^\bullet(t)=\varnothing$ and ${R'}^\bullet(t')=\varnothing$). By structural induction, every formula evaluated under the strict interpretation ($G^\ast,H^\ast$) has the same truth value at $t$ and $t'$ (the only subformulas are $\bot$, atoms, boolean connectives, and $G^\ast,H^\ast$; all are interpreted identically). However, the formula $Gp$ under the standard interpretation is false at $t$ (because $tRt$ and $p$ is false) but true at $t'$ (vacuously, as $R'(t')=\varnothing$).

\smallskip
\noindent\emph{(ii) The strict interpretation cannot be captured by the standard interpretation.}
Let $\mathcal{M}$ be the same model as in (i) (a single reflexive point $t$ with all atoms false). Let $\mathcal{N}^\prime$ be a model consisting of two points $t',u'$ such that $t'R't'$, $u'R'u'$, and $t'R'u'$, with all atoms false at both points. A straightforward structural induction shows that $t$ (in $\mathcal{M}$) and $t'$ (in $\mathcal{N}^\prime$) satisfy exactly the same formulas under the standard interpretation. Indeed, the only difference is the extra point $u'$ in $\mathcal{N}^\prime$, but since all atoms are false there, its presence does not affect the truth of any formula. Consequently, $t$ and $t'$ are indistinguishable for $L(G,H)$. On the other hand, they differ under the strict interpretation: $G^\ast p$ is true at $t$ (since $R^\bullet(t)=\varnothing$) but false at $t'$ because $u'$ is a strict successor of $t'$ and $p$ is false there. \qed
\end{proof}

\subsection{Definability under the strict interpretation}\label{subsec:strict-results-general}

We now examine whether adopting the \emph{strict interpretation} ($G^*, H^*$) of our temporal operators alters the definability results in the multiflow setting. 

Regarding positive results, the properties of totality and surjectivity remain definable in both $\mathsf{LO}$ and $\mathsf{sLO}$ frameworks, as shown in Theorem~\ref{tot-surj}. 

As for the negative results, the counterexamples in Figures~\ref{fig:total}--\ref{fig:dec} and~\ref{fig:inj_SPRE2} consist of strict orders where $R_w = R^\bullet_w$. In these models, the standard interpretation and the strict interpretation are semantically equivalent. Although Convention~(RC) allows these frames to represent reflexive closures, the strict interpretation ensures that $G^*,H^*$ quantify exclusively over the underlying relation $R^\ast_w$, ignoring any reflexive loops. This ensures that it evaluates the same structure as in the original strict counterexamples.

Consequently, the surjective p-morphisms established for the standard interpretation remain valid for the strict interpretation. All undefinability results therefore transfer directly across every order type, leaving Theorem~\ref{main} and Table~\ref{tab:definability-summary} unchanged. The expressive limitations are inherent to the multiflow structure, not the choice of Priorean or strict semantics.

\section{Structural Refinements: Minimal Functional Frames}
\label{sec:minimal}

The findings in Section~\ref{sec:ordered-multiflow} demonstrate that definability within the multiflow setting is severely constrained by its structural complexity. This section investigates whether these limitations can be overcome by imposing a simpler structure on the frames while keeping the language $L_{T\times W}$ fixed.

To this end, we introduce \emph{minimal functional frames} (the $\Or^2$ family)---structures restricted to at most two flows and a single accessibility function. This refinement effectively transforms our functional operators $\square$ and $\lozenge$ from second-order operators (which quantify over an arbitrary family of functions) into first-order ones, as they now refer to a single, fixed mapping. 

Our analysis follows a two-step approach:
\begin{enumerate}
    \item We examine definability in minimal frames using the standard pair $(G, H)$, establishing what gains can be achieved by structural simplification alone.
    \item We evaluate the same minimalist environment using the strict pair $(G^\ast, H^\ast)$.
\end{enumerate}
This allows us to determine whether the choice of temporal  interpretation---which proved irrelevant in the general multiflow setting---becomes a decisive factor for expressivity once structural noise is eliminated.

\subsection{Definability in Minimal Frames under the Standard Temporal Interpretation}
\label{subsec:minimal-prior}

We now examine minimal functional frames, a setting where definability is recovered for a significant range of properties. By reducing structural complexity to the $\Or^2$ family, these frames isolate functional behavior and allow for a focused analysis of the interaction between order types and the \emph{standard} temporal operators $G$ and $H$. While the general multiflow setting restricted definability to just two properties---totality and surjectivity---and only within linear orders (strict or not), the minimal settingt enables the definability of several additional properties.

\begin{definition}
A functional frame $\Sigma = (W, \mathcal{T}, \mathcal{F})$ is \emph{minimal} if $|W| \leq 2$ and $|\mathcal{F}| \leq 1$. If $\Sigma$ is minimal, then $\mathcal{M} = (\Sigma, h)$ is a \emph{minimal functional model}. For any order type $\Or \in \{\PRE, \sPRE, \PO, \LO, \sLO\}$, we denote by $\Or^2$ the class of minimal functional frames of type $\Or$.
\end{definition}

\medskip
\noindent
In what follows we assume the frame contains at least one function; the case $\mathcal{F}=\varnothing$ does not affect the definability results. Thus $\mathcal{F}_w = \{f_{ww'}\}$ for some $w' \in W$, and the correspondences 
from Lemma~\ref{lem:image-modal-unified} simplify to:

\begin{corollary}\label{cor:singleton-Fw}
Let $\Sigma=(W,\mathcal T,\mathcal{F})$ be a functional frame with $\mathcal{F}_w = \{f_{ww'}\}$. For every $X \subseteq T_w$, interpretation $h$, and formula $A$, the following equivalences hold:
\begin{itemize}
\item For $\mathsf{sPRE}$ and $\mathsf{sLO}$:
\begin{enumerate}
  \item $X \subseteq h(\square A)$ if and only if $f_{ww'}(X) \subseteq h(A)$.
  \item $f_{ww'}(X)\!\uparrow^{\ast} \subseteq h(A)$ if and only if $f_{ww'}(X) \subseteq h(GA)$.
  \item $f_{ww'}(X)\!\downarrow^{\ast} \subseteq h(A)$ if and only if $f_{ww'}(X) \subseteq h(HA)$.
\end{enumerate}
\item For $\mathsf{PRE}, \mathsf{PO}$, and $\mathsf{LO}$: The same equivalences hold, but in (2) and (3) replace $\uparrow^\ast$ by $\uparrow$ and $\downarrow^{\ast}$ by $\downarrow$.
\end{itemize}
\end{corollary}

\begin{theorem}[Definability in Minimal Functional Frames]\label{prop:minimal-definability}
The definability of functional properties in minimal frames depends on the underlying order type:
\begin{enumerate}
    \item Each of \emph{totality}, \emph{non-totality}, \emph{surjectivity}, and \emph{constancy} is definable only in $\LO^2$ and $\sLO^2$.
    \item \emph{Injectivity} is definable only in $\sLO^2$.
    \item \emph{Monotonicity} and \emph{antitonicity} are definable in all minimal order types.
    \item Each of \emph{strict monotonicity} and \emph{strict antitonicity} is definable only in $\sPRE^2$ and $\sLO^2$.
\end{enumerate}
\end{theorem}

\begin{proof}
The theorem follows from Lemmas~\ref{lem:total-surj-const} to~\ref{lem:strictness} below. \qed
\end{proof}

\begin{lemma}[Totality, Non-totality, Surjectivity, and Constancy]\label{lem:total-surj-const}
Each of totality, non-totality, surjectivity, and constancy is definable only in $\LO^2$ and $\sLO^2$.
\end{lemma}

\begin{proof}
We first address definability, then undefinability.

\textsc{Definability.}

\begin{itemize}
    \item \emph{Totality.} As in Theorem~\ref{tot-surj}, the formula $(\Tot)$ defines $\mathbb{K}^{\sLO^2}_{\tot}$ in $\mathbb{K}^{\sLO^2}$; its simplified version $(\Tot)^\circ$ defines $\mathbb{K}^{\LO^2}_{\tot}$ in $\mathbb{K}^{\LO^2}$. (Note that $(\Tot)$ also works in $\LO^2$, where the conjunct $p$ is redundant, but $(\Tot)^\circ$ does not hold in $\sLO^2$.)
        \item \emph{Surjectivity.} As in Theorem~\ref{tot-surj}, the formula $(\Surj)$ defines $\mathbb{K}^{\mathsf{sLO}^2}_{\surj}$ in $\mathbb{K}^{\mathsf{sLO}^2}$ and $\mathbb{K}^{\mathsf{LO}^2}_{\surj}$ in $\mathbb{K}^{\mathsf{LO}^2}$.
\end{itemize}

The verification follows Theorem~\ref{tot-surj}, adapted via Corollary~\ref{cor:singleton-Fw} to the minimal setting. Below, we provide the details for the specific cases of non-totality, constancy, and injectivity, introducing their defining formulas in each section.

\medskip\noindent
\emph{Non-totality.} 
The following formula defines $\mathbb{K}^{\mathsf{sLO}^2}_{\ntot}$ in $\mathbb{K}^{\mathsf{sLO}^2}$:
\[
(\Ntot)^2:\ P\square\bot \lor \square\bot \lor F\square\bot.
\]
($\Rightarrow$) Assume $\Sigma \in \mathbb{K}^{\mathsf{sLO}^2}_{\ntot}$ and let $f_{ww'}$ be its unique function. There exists $t_w$ with $f_{ww'}(\{t_w\}) = \varnothing$. Let $(\Sigma, h)$ be any model,  then $t_w \in h(\square\bot)$. In a strict linear order, every $t'_w$ satisfies $t'_w < t_w$, $t'_w = t_w$, or $t'_w > t_w$. Thus, $t'_w \in h(P\square\bot)$, $t'_w \in h(\square\bot)$, or $t'_w \in h(F\square\bot)$, respectively. Hence $(\Ntot)^2$ is valid in $\Sigma$.

($\Leftarrow$) If $\Sigma \notin \mathbb{K}^{\mathsf{sLO}^2}_{\ntot}$, then for every $t_w$, $f_{ww'}(\{t_w\}) \neq \varnothing$. Consequently, for any model $(\Sigma, h)$, we have $t_w \in h(\neg\square\bot)$ for all $t_w$. Since the frame is a strict linear order, $t_w\!\downarrow^\ast \cup\, \{t_w\} \,\cup t_w\!\uparrow^\ast = T_w \subseteq h(\neg\square\bot)$. It follows that $t_w \notin h(P\square\bot \lor \square\bot \lor F\square\bot)$, and therefore $(\Ntot)^2$ is invalid in $\Sigma$.

\medskip
\noindent

\emph{Constancy.}
The following formula defines $\mathbb{K}^{\mathsf{sLO}^2}_{\con}$ in $\mathbb{K}^{\mathsf{sLO}^2}$:
\[
(\Con)^2:\ \lozenge p \to (H\square p \land G\square p).
\]
Let $\Sigma \in \mathbb{K}^{\mathsf{sLO}^2}$ and let $f_{ww'}$ its unique function. By Theorem~\ref{fundamentos}(6), $f_{ww'}$ is constant iff for every $t_w \in \mathrm{Dom}(f_{ww'})$,
\[
f_{ww'}(t_w\!\downarrow^\ast) \cup f_{ww'}(t_w\!\uparrow^\ast) \subseteq \{f_{ww'}(t_w)\}. \qquad (\con)
\]

($\Rightarrow$) Suppose $\Sigma \in \mathbb{K}^{\mathsf{sLO}^2}_{\con}$.  Let $(\Sigma, h)$ be any model and suppose $t_w \in h(\lozenge p)$. Since $f_{ww'}$ is constant, by $(\con)$ we have
\[
f_{ww'}(t_w\!\downarrow^\ast) \cup f_{ww'}(t_w\!\uparrow^\ast) \subseteq \{f_{ww'}(t_w)\} \subseteq h(p). \tag{1}
\]
Applying Corollary~\ref{cor:singleton-Fw}(1) with $X = t_w\!\downarrow^\ast \cup\, t_w\!\uparrow^\ast$ yields
\[
t_w\!\downarrow^\ast \cup\,  t_w\!\uparrow^\ast \subseteq h(\square p). \tag{2}
\]
From (2), it follows that $t_w\!\downarrow^\ast \subseteq h(\square p)$ and $t_w\!\uparrow^\ast \subseteq h(\square p)$. Thus, Corollary~\ref{cor:singleton-Fw}(2) and (3) yield $t_w \in h(H\square p)$ and $t_w \in h(G\square p)$, respectively. Hence
\[
t_w \in h(H\square p \land G\square p),
\]
and therefore $(\Con)^2$ is valid in $\Sigma$.

($\Leftarrow$) Suppose $\Sigma \notin \mathbb{K}^{\mathsf{sLO}^2}_{\con}$. Because $f_{ww'}$ is not constant, Theorem~\ref{fundamentos}(6) gives a point $t_w$ such that $(\con)$ fails. Consider a model $(\Sigma, h)$, where  $h(p) := \{f_{ww'}(t_w)\}$. Then $t_w \in h(\lozenge p)$ (since $f_{ww'}(t_w) \in h(p)$). By the failure of $(\con)$, we have
\[
f_{ww'}(t_w\!\downarrow^\ast) \cup f_{ww'}(t_w\!\uparrow^\ast) \not\subseteq \{f_{ww'}(t_w)\}.  \tag{3}
\]
Applying Corollary~\ref{cor:singleton-Fw}(1) yields $t_w\!\downarrow^\ast \cup \, t_w\!\uparrow^\ast \not\subseteq h(\square p)$. Thus, by Corollary~\ref{cor:singleton-Fw}(2) and (3), $t_w \notin h(H\square p)$ or $t_w \notin h(G\square p)$, whence $t_w \notin h(H\square p \land G\square p)$. Therefore $(\Con)^2$ is invalid in $\Sigma$.

\medskip
\textsc{ii. Undefinability.}
For totality, non-totality, and surjectivity, the counterexamples from the general setting (Figures~\ref{fig:total} and~\ref{fig:surj}) remain valid in the minimal setting. They witness undefinability in $\mathsf{sPRE}^2$; by Convention~(RC), these counterexamples also establish undefinability in $\mathsf{PRE}^2$ and $\mathsf{PO}^2$.

Constancy requires a different approach because surjective p-morphisms preserve constancy, as the following claim shows.

\begin{claim}
Constancy is preserved by surjective p-morphisms for partial functions.
\end{claim}
\emph{Proof of the Claim}
Let \(f \colon P_1 \rightharpoonup P_2\), \(f' \colon P'_1 \rightharpoonup P'_2\) with \(\operatorname{Dom}(f) \neq \varnothing\), and let \(h_1 \colon P_1 \to P'_1\), \(h_2 \colon P_2 \to P'_2\) satisfy:
\begin{enumerate}
    \item \(h_1|_{\operatorname{Dom}(f)}\) is onto \(\operatorname{Dom}(f')\);
    \item \(f'(h_1(t)) = h_2(f(t))\) for every \(t \in \operatorname{Dom}(f)\).
\end{enumerate}
Assume \(f\) constant but \(f'\) not. Then there exist \(a'_1, a'_2 \in \operatorname{Dom}(f')\) with \(f'(a'_1) \neq f'(a'_2)\). By (1) choose \(a_1, a_2 \in \operatorname{Dom}(f)\) with \(h_1(a_i) = a'_i\) (\(i=1,2\)). Condition (2) gives \(f'(a'_i) = h_2(f(a_i))\). Since \(f\) constant, \(f(a_1)=f(a_2)=:y\). Thus \(h_2(y)=f'(a'_1)=f'(a'_2)\), contradiction. Hence \(f'\) must be constant.

\medskip
Since the p-morphism method fails to separate these frames, we establish their indistinguishability by a direct valuation mapping. For any formula falsifiable in a model over \(\Sigma'\), we construct a model over \(\Sigma\) that falsifies the same formula by copying the valuation pointwise, as described below.

Consider the following two minimal frames, where temporal relations are empty (reflexively closed for non-strict variants):
\[
\begin{aligned}
\Sigma &: \; T_w = \{1_w, 2_w\},\; T_v = \{3_v\},\; \mathcal{F} = \{f_{wv}\},\; f_{wv}(1_w) = f_{wv}(2_w) = 3_v,\\[2pt]
\Sigma' &: \; T'_{w'} = \{4_{w'}, 5_{w'}\},\; T'_{v'} = \{6_{v'}, 7_{v'}\},\; \mathcal{F}' = \{f'_{w'v'}\},\\
      &\qquad f'_{w'v'}(4_{w'}) = 6_{v'},\; f'_{w'v'}(5_{w'}) = 7_{v'}.
\end{aligned}
\]
We now show that if a formula \(A\) is falsifiable in a model over \(\Sigma'\), then it is also falsifiable in a model over \(\Sigma\). Let \(M_2 = (\Sigma', h')\) with \(M_2, x' \not\models A\). Build \(M_1 = (\Sigma, h)\) by case analysis:

We build \(M_1 = (\Sigma, h)\) by transferring the valuation pointwise. 
For every propositional variable \(p\) we define \(h(p)\) as follows, according to the position of the falsifying point \(x'\) in \(M_2\):

\begin{itemize}
\item If \(x' = 4_{w'}\): set \(1_w \in h(p)\) iff \(4_{w'} \in h'(p)\), and \(3_v \in h(p)\) iff \(6_{v'} \in h'(p)\).
\item If \(x' = 5_{w'}\): set \(1_w \in h(p)\) iff \(5_{w'} \in h'(p)\), and \(3_v \in h(p)\) iff \(7_{v'} \in h'(p)\).
\item If \(x' = 6_{v'}\) or \(x' = 7_{v'}\): set \(3_v \in h(p)\) iff \(x' \in h'(p)\).
\end{itemize}

Values at the remaining coordinates may be chosen arbitrarily. By structural induction on \(A\), \(M_1\) falsifies \(A\) at the corresponding coordinate. Consequently, no \(L_{T\times W}\)-formula can distinguish \(\Sigma\) from \(\Sigma'\). Since \(\Sigma\) has a constant function and \(\Sigma'\) has not, constancy is not definable in \(\sPRE^2\). The same construction with reflexive closures yields undefinability in \(\PRE^2\) and \(\PO^2\).\qed
\end{proof}

\begin{lemma}[Injectivity]\label{lem:injectivity}
Injectivity is definable only in $\sLO^2$.
\end{lemma}

\begin{proof}
The formula
\[
(\Inj)^2:\ \lozenge(Hp \land Gp) \to (H\square p \land G\square p)
\]
defines $\mathbb{K}^{\sLO^2}_\inj$ in $\mathbb{K}^{\sLO^2}$.

\medskip\noindent
\emph{Strict linear orders} ($\sLO^{2}$).
Let $\Sigma \in \mathbb{K}^{\sLO^2}$ and let $f_{ww'}$ be its unique function. By Theorem~\ref{fundamentos}(2) $f_{ww'}$ is injective iff for every $t_w \in \mathrm{Dom}(f_{ww'})$,
\[
f_{ww'}(t_w\!\downarrow^{*})\cup f_{ww'}(t_w\!\uparrow^{*}) \subseteq 
f_{ww'}(t_w)\!\downarrow^{*}\cup\, f_{ww'}(t_w)\!\uparrow^{*}. \qquad (\inj)
\]

($\Rightarrow$) Assume $\Sigma \in \mathbb{K}^{\sLO^2}_\inj$.  Let  $(\Sigma, h)$ be any model and let  $t_w \in h(\lozenge(Hp \land Gp))$. Then $f_{ww'}(t_w) \in h(Hp) \cap h(Gp)$, which implies
\[
f_{ww'}(t_w)\!\downarrow^{\ast} \cup\, f_{ww'}(t_w)\!\uparrow^{\ast} \subseteq h(p). \tag{1}
\]
By $(\inj)$ and (1), we obtain $f_{ww'}(t_w\!\downarrow^{\ast}) \cup f_{ww'}(t_w\!\uparrow^{\ast}) \subseteq h(p)$. Applying Corollary~\ref{cor:singleton-Fw}(1) with $X = t_w\!\downarrow^{\ast} \cup\, t_w\!\uparrow^{\ast}$ yields
\[
t_w\!\downarrow^{\ast} \cup\, t_w\!\uparrow^{\ast} \subseteq h(\square p). \tag{2}
\]
Hence $t_w\!\downarrow^{\ast} \subseteq h(\square p)$ and $t_w\!\uparrow^{\ast} \subseteq h(\square p)$. By Corollary~\ref{cor:singleton-Fw}(2) and (3), we obtain $t_w \in h(H\square p)$ and $t_w \in h(G\square p)$. Therefore $t_w \in h(H\square p \land G\square p)$, hence $(\Inj)^2$ is valid in $\Sigma$.

($\Leftarrow$) Suppose $\Sigma \notin \mathbb{K}^{\sLO^2}_\inj$. Then $f_{ww'}$ is not injective, so by Theorem~\ref{fundamentos}(2) there exists $t_w \in \mathrm{Dom}(f_{ww'})$ such that $(\inj)$ fails. Define a model $(\Sigma, h)$, where $h(p) := f_{ww'}(t_w)\!\downarrow^{\ast} \cup\, f_{ww'}(t_w)\!\uparrow^{\ast}$. Then $t_w \in h(\lozenge(Hp \land Gp))$ but
\[
f_{ww'}(t_w\!\downarrow^{\ast}) \cup f_{ww'}(t_w\!\uparrow^{\ast}) \not\subseteq h(p),
\]
so by Corollary~\ref{cor:singleton-Fw}(1), $t_w\!\downarrow^{\ast} \cup\, t_w\!\uparrow^{\ast} \not\subseteq h(\square p)$. Hence $t_w \notin h(H\square p)$ or $t_w \notin h(G\square p)$, i.e. $t_w \notin h(H\square p \land G\square p)$. Thus $(\Inj)^2$ is invalid in $\Sigma$.

\medskip\noindent
\emph{Undefinability elsewhere.}
Figure~\ref{fig:inj_SPRE2} exhibits a surjective p-morphism witnessing
indefinability in $\sPRE^2$; by Convention~(RC) its reflexive closure
yields indefinability in $\PRE^2$ and $\PO^2$.
Figure~\ref{fig:inj_P-PO-LO} shows indefinability in $\LO^2$. \qed
\end{proof}

\begin{lemma}[Monotonicity and Antitonicity]\label{lem:monotonicity}
Each of monotonicity and antitonicity is definable in all minimal order types.
\end{lemma}

\begin{proof}
\emph{Monotonicity.}   We treat non-strict and strict orders separately, as they require different formulas.

\medskip\noindent
\emph{Non-strict orders }($\PRE^{2}$, $\PO^{2}$, $\LO^{2}$).
Let $\Sigma \in \mathbb{K}^{\Or^2}$ with $\Or \in \{\PRE, \PO, \LO\}$ and let $f_{ww'}$ be its unique function.  
By Theorem~\ref{local}(1), $f_{ww'}$ is increasing iff for every $t_w \in \mathrm{Dom}(f_{ww'})$,
\[
f_{ww'}(t_w\!\uparrow) \subseteq f_{ww'}(t_w)\!\uparrow. \tag{$\inc$}
\]

We show that $(\Inc)^{2}: \lozenge Gp \to G\square p$ defines monotonicity in these classes.

($\Rightarrow$) Assume $\Sigma \in \mathbb{K}^{\Or^2}_{\inc}$.  
Let $(\Sigma,h)$ be any model and suppose $t_w \in h(\lozenge Gp)$.  
Then $f_{ww'}(t_w) \in h(Gp)$, which in reflexive orders means $f_{ww'}(t_w)\!\uparrow\,\subseteq h(p)$.  
From this and $(\inc)$ we obtain $f_{ww'}(t_w\!\uparrow) \subseteq h(p)$.  
Applying Corollary~\ref{cor:singleton-Fw}(1) with $X = t_w\!\uparrow$ yields $t_w\!\uparrow \, \subseteq h(\square p)$, and Corollary~\ref{cor:singleton-Fw}(3) then gives $t_w \in h(G\square p)$.  
Hence $(\Inc)^{2}$ is valid in $\Sigma$.

($\Leftarrow$) Assume $\Sigma \notin \mathbb{K}^{\Or^2}_{\inc}$.  
Then $(\inc)$ fails at some $t_w \in \mathrm{Dom}(f_{ww'})$.  
Define $h(p) := f_{ww'}(t_w)\!\uparrow$.  
Clearly $t_w \in h(\lozenge Gp)$.  
But $f_{ww'}(t_w\!\uparrow) \not\subseteq h(p)$ by failure of  $(\inc)$ and definition of $h$, so by Corollary~\ref{cor:singleton-Fw}(1) we have $t_w\!\uparrow \not\subseteq h(\square p)$, whence $t_w \notin h(G\square p)$.  
Thus $(\Inc)^2$ is invalid in $\Sigma$.

\medskip\noindent
\emph{Strict orders} ($\sPRE^2$, $\sLO^2$). 
Let $\Sigma \in \mathbb{K}^{\Or^2}$ with $\Or \in \{\sPRE, \sLO\}$.  
For strict orders we replace $\uparrow$ by $\uparrow^{*}$ throughout.  
The characterisation becomes $$f_{ww'}(t_w\!\uparrow^{*}) \subseteq f_{ww'}(t_w)\!\uparrow^{*}$$and the defining formula is
\[
(\Inc)^{2}_s:\ \lozenge(p \land Gp) \to G\square p .
\]
The proof proceeds exactly as above, noting that $f_{ww'}(t_w) \in h(Gp)$ now requires both $f_{ww'}(t_w) \in h(p)$ \emph{and} $f_{ww'}(t_w)\!\uparrow^{*} \subseteq h(p)$.

\medskip\noindent
\emph{Antitonicity.}  
This follows by the same argument, replacing $\uparrow$ with $\downarrow$ in the consequent of $(\inc)$. The characterisation becomes $f_{ww'}(t_w\!\uparrow) \subseteq f_{ww'}(t_w)\!\downarrow$ (Theorem~\ref{local}(2)), yielding the formulas $(\Dec)^2: \lozenge Gp \to H\square p$ for non-strict orders and $(\Dec)^{2}_s: \lozenge(p \land Gp) \to H\square p$ for strict orders.

\medskip
Since the argument covers every minimal order type, monotonicity and antitonicity are definable in all of them.
\end{proof}

\begin{lemma}[Strict monotonicity and strict antitonicity]\label{lem:strictness}
Each of strict monotonicity and strict antitonicity is definable only in $\sPRE^2$ and $\sLO^2$.
\end{lemma}

\begin{proof} 
\emph{Definability in $\sPRE^2$ and $\sLO^2$.}
The formulas
\[
(\mathit{Inc})^{2}:\ \lozenge Gp \to G\square p 
\qquad\text{and}\qquad 
(\mathit{Dec})^{2}:\ \lozenge Gp \to H\square p,
\]
which define ordinary monotonicity and antitonicity in non-strict orders, also define their strict counterparts in strict orders. In a strict order, $(\mathit{Inc})^{2}$ expresses precisely the condition
$
f_{ww'}(t_w\!\uparrow^{*}) \subseteq f_{ww'}(t_w)\!\uparrow^{*},
$
which is the characterisation of strict monotonicity (Theorem~\ref{fundamentos}(4)). The verification is identical to the proof of Lemma~\ref{lem:monotonicity} for the non-strict case, replacing $\uparrow,\downarrow$ everywhere by $\uparrow^{*},\downarrow^{*}$.

\medskip
\noindent
\emph{Undefinability elsewhere.}
Figure~\ref{fig:inj_P-PO-LO} contains a surjective homomorphism that shows strict monotonicity 
is not definable in $\PRE^{2}$, $\PO^{2}$ or $\LO^{2}$; 
Figure~\ref{fig:strict-antitonicity_P-PO-LO} does the same for strict antitonicity. \qed
\end{proof}

\begin{table}[ht]
\centering
\small
\caption{Definability of functional properties in minimal frames ($O^2$ family) under $(G,H)$.}
\label{tab:definability-minimal-frames}
\renewcommand{\arraystretch}{1.2} 
\setlength{\tabcolsep}{0pt}      
\begin{tabular*}{0.85\textwidth}{@{\extracolsep{\fill}} l ccccc @{}}
\toprule
\textbf{Property} & \textbf{PRE$^2$} & \textbf{sPRE$^2$} & \textbf{PO$^2$} & \textbf{LO$^2$} & \textbf{sLO$^2$} \\
\midrule
Totality            & No  & No  & No  & Yes & Yes \\
Non-totality        & No  & No  & No  & Yes & Yes \\
Injectivity         & No  & No  & No  & No  & Yes \\
Surjectivity        & No  & No  & No  & Yes & Yes \\
Monotonicity        & Yes & Yes & Yes & Yes & Yes \\
Strict monotonicity & No  & Yes & No  & No  & Yes \\
Antitonicity        & Yes & Yes & Yes & Yes & Yes \\
Strict antitonicity & No  & Yes & No  & No  & Yes \\
Constancy           & No  & No  & No  & Yes & Yes \\
\bottomrule
\end{tabular*}
\end{table}

\subsection{Definability Improvements under the Strict Temporal Interpretation}
\label{subsec:strict-improvements}

The strict interpretation $G^\ast, H^\ast$ significantly expands the definability landscape. By disregarding reflexive loops, it breaks certain surjective homomorphisms that preserved truth for the standard interpretation $G, H$, turning them into non-bisimulations.

\begin{lemma}[Persistence of (un)definability]\label{lem:strict-persistence}
Over minimal functional frames under the strict interpretation ($G^\ast, H^\ast$):
\begin{enumerate}
  \item Totality, non-totality, surjectivity, monotonicity, antitonicity, and constancy are definable in the same order types as under the standard interpretation. Specifically:
  \begin{itemize}
    \item[(a)] In $\LO^2$ and $\sLO^2$: 
      \[
        \begin{aligned}
        &\text{Totality: } && \square(H^\ast p \land p \land G^\ast p)\to (H^\ast\square p \land G^\ast\square p) \\
        &\text{Non-totality: } && P^\ast\square\bot \lor \square\bot \lor F^\ast\square\bot \\
        &\text{Surjectivity: } && (H^\ast\square p \land G^\ast\square p)\to \square(H^\ast p \land G^\ast p) \\
        &\text{Constancy: } && \lozenge p \to (H^\ast\square p \land G^\ast\square p)
        \end{aligned}
      \]
    \item[(b)] In $\PRE^2$, $\PO^2$, $\LO^2$, $\sPRE^2$, $\sLO^2$:
      \[
        \begin{aligned}
        &\text{Monotonicity: } && \lozenge(p \land G^\ast p) \to G^\ast\square p \\
        &\text{Antitonicity: } && \lozenge(p \land G^\ast p) \to H^\ast\square p
        \end{aligned}
      \]
  \end{itemize}
  \item Totality, non-totality, injectivity, surjectivity, and constancy remain undefinable in $\PRE^2$, $\sPRE^2$, and $\PO^2$.
\end{enumerate}
\end{lemma}

\begin{proof}
\noindent\textit{Proof sketch.} In strict orders, the standard and strict interpretations coincide semantically. In reflexive orders, the strict interpretation ignores reflexive loops and evaluates only the strict part $\uparrow^\ast,\downarrow^\ast$. Hence, from the perspective of the strict interpretation, all orders are effectively assimilated to their strict counterparts: $\PRE^2$ and $\PO^2$ behave like $\sPRE^2$, and $\LO^2$ behaves like $\sLO^2$.

Definability transfers directly from the standard to the strict setting:
\begin{itemize}
    \item Totality, surjectivity, constancy, and non-totality: the proofs for $\sLO^2$ under the standard interpretation (Lemma~\ref{lem:total-surj-const}) apply unchanged to $\LO^2$ and $\sLO^2$ under the strict interpretation.
    \item Monotonicity and antitonicity: the proofs for $\sPRE^2$ and $\sLO^2$ under the standard interpretation (Lemma~\ref{lem:monotonicity}) extend to all minimal order types under the strict interpretation.
\end{itemize}

Undefinability in $\PRE^2$, $\sPRE^2$, and $\PO^2$ follows from the same counterexamples used for the standard interpretation (Figures~\ref{fig:total}, \ref{fig:surj}, \ref{fig:inj_SPRE2}) together with Convention~(RC). For constancy, the semantic equivalence argument from Theorem~\ref{prop:minimal-definability} applies unchanged to the strict interpretation in $\sPRE^2$ and, via Convention~(RC), to $\PRE^2$ and $\PO^2$. \qed
\end{proof}

\begin{lemma}[Definability gains with the strict interpretation]\label{lem:strict-gains}
Over minimal functional frames, the strict interpretation ($G^*, H^*$) makes definable properties previously undefinable under the standard interpretation ($G, H$):
\begin{enumerate}
  \item Injectivity is definable in $\LO^2$ by
        \[ (Inj)_*^2:\; \lozenge(H^*p \land G^*p) \to (H^*\square p \land G^*\square p). \]
  \item Strict monotonicity and strict antitonicity are definable in $\PRE^2$, $\PO^2$, and $\LO^2$ by
        \[ (Inc)_*^2:\; \lozenge G^*p \to G^*\square p, \qquad (Dec)_*^2:\; \lozenge G^*p \to H^*\square p. \]
\end{enumerate}
\end{lemma}

\begin{proof}

The algebraic characterizations for injectivity (Theorem~\ref{fundamentos}(2)) and strict monotonicity (Theorem~\ref{fundamentos}(5)) involve only the strict intervals $\uparrow^*$ and $\downarrow^*$. Under the strict interpretation, $G^*$ and $H^*$ quantify precisely over these strict intervals, even in reflexive orders. Therefore, the verification follows exactly the same pattern as in Lemma~\ref{lem:injectivity} for injectivity in $\mathsf{sLO}^2$ and Lemma~\ref{lem:strictness} for strict monotonicity in $\mathsf{sPRE}^2$ and $\mathsf{sLO}^2$, replacing $G,H$ with $G^*,H^*$ throughout. The explicit algebraic manipulations are omitted to avoid repetition; the reader may consult the aforementioned lemmas for the detailed step-by-step derivation.
\end{proof}

\begin{theorem}[Definability under the strict interpretation]\label{thm:strict-complete-classification}
Over minimal functional frames under the strict interpretation ($G^*, H^*$):
\begin{itemize}
\item Totality, non-totality, injectivity, surjectivity, and constancy are definable only in linear orders ($\LO^2$ and $\sLO^2$).
    \item Monotonicity, antitonicity, strict monotonicity, and strict antitonicity are definable in all order types.
    
\end{itemize}
\end{theorem}
\begin{proof}
Lemma~\ref{lem:strict-persistence}(1) shows totality, non-totality, surjectivity, and constancy are definable in $\LO^2$ and $\sLO^2$, while Lemma~\ref{lem:strict-gains}(1) adds injectivity to this linear-order class.

From Lemma~\ref{lem:strict-persistence}(1) and Lemma~\ref{lem:strict-gains}(2), monotonicity, antitonicity, and their strict variants are definable in all minimal order types. 

Finally, Lemma~\ref{lem:strict-persistence}(2) establishes that these five properties remain undefinable in $\PRE^2$, $\sPRE^2$, and $\PO^2$. \qed
\end{proof}

\medskip
The definability results under the strict interpretation are summarized in Table~\ref{tab:definability-minimal-frames-strict}.
\begin{table}[ht]
\centering
\small
\caption{Definability of functional properties in minimal frames under the strict interpretation (Yes = definable)}
\label{tab:definability-minimal-frames-strict}
\renewcommand{\arraystretch}{1.25}
\setlength{\tabcolsep}{0pt}
\begin{tabular*}{0.85\textwidth}{@{\extracolsep{\fill}} l ccccc @{}}
\toprule
\textbf{Property} & \textbf{PRE$^2$} & \textbf{sPRE$^2$} & \textbf{PO$^2$} & \textbf{LO$^2$} & \textbf{sLO$^2$} \\
\midrule
Totality            & No  & No  & No  & Yes & Yes \\
Non-totality        & No  & No  & No  & Yes & Yes \\
Injectivity         & No  & No  & No  & Yes & Yes \\
Surjectivity        & No  & No  & No  & Yes & Yes \\
Monotonicity        & Yes & Yes & Yes & Yes & Yes \\
Strict monotonicity & Yes & Yes & Yes & Yes & Yes \\
Antitonicity        & Yes & Yes & Yes & Yes & Yes \\
Strict antitonicity & Yes & Yes & Yes & Yes & Yes \\
Constancy           & No  & No  & No  & Yes & Yes \\
\bottomrule
\end{tabular*}
\end{table}

\subsection{Relation to Other Frameworks: Indexed Languages and Uniform Domains}\label{subsec:relation-frameworks}
Our results can be viewed in light of two existing frameworks: indexed modal languages \cite{Burrieza2009} and uniform-domain frames \cite{Burrieza2003}.

\subsubsection{Relation to indexed modal languages: a structural insight.}

Indexed languages achieve their expressivity by employing indexed modal operators ($[i], \langle i \rangle$, where $i$ is an index for a flow) to distinguish the \emph{range} of functions that are otherwise indistinguishable in an arbitrary ordered multiflow setting. By attaching these indices to the connectives, the language allows specific accessibility functions to be uniquely identified, as each index $i$ points to the destination flow of a single function $f_{wi}$. This strategy was systematically analyzed for strict linear orders ($\sLO$) and the same group of functional properties as in this paper in \cite{Burrieza2009}.

The fact that our definability patterns for minimal frames ($\Or^{2}$) coincide precisely with those established for indexed languages is not accidental. If a functional property $P$ holds in an indexed frame, it must hold for each function $f_{wi}$ individually. Consequently, any definability argument---including those in \cite{Burrieza2009}---necessarily focuses on a single function and the two flows it connects, rendering the rest of the multiflow structure logically inert.

This insight provides a direct translation between our results and those for indexed languages:
\begin{itemize}
  \item \emph{Definability:} If a formula $A$ of $L_{T\times W}$ defines a property $P$ in the minimal class $\Or^{2}$ (as established in Section~\ref{sec:minimal}), then $P$ is also defined in the class of indexed frames of type $\Or$ by the indexed schema 
\[
\{A(i) \mid i \in I\},
\]
where each $A(i)$ is obtained by replacing $\square$ with $[i]$ and $\lozenge$ with $\langle i\rangle$. The verification is identical to that of Theorem~\ref{prop:minimal-definability}, simply replacing $f_{ww'}$ with $f_{wi}$.

    \item \emph{Undefinability:} Every undefinability counterexample in Appendix (see Figs.~\ref{fig:total},~\ref{fig:surj},~\ref{fig:inj_SPRE2},~\ref{fig:inj_P-PO-LO}, and~\ref{fig:strict-antitonicity_P-PO-LO},  complemented by the (RC) convention where applicable) involves only two flows and a single function. 
          Interpreting the target flow as an index $i$ yields an immediate counterexample for the indexed setting, via the same surjective homomorphism.
\end{itemize}
This bidirectional correspondence ensures that the results established in Tables~\ref{tab:definability-minimal-frames} and~\ref{tab:definability-minimal-frames-strict} are exhaustive for the group of functional properties analyzed in this study. 

Regarding the choice of sources, we take \cite{Burrieza2009} as our primary reference because it is the first indexed work where functional properties appear in isolation, matching our current approach. Earlier treatments \cite{Burrieza2002indexed} invariably combined them with totality, preventing an independent analysis of each property. Furthermore, while subsequent work introduced double-indexing to achieve completeness for surjective functions (see \cite{Burrieza2017}), such refinements are redundant for definability. Since modal evaluation begins at a specific ``actual'' flow, the destination index alone uniquely identifies the function, making additional domain indexing unnecessary for the purposes of this study.

Thus, our minimal-frame approach reveals that adding explicit indices to the operators (syntactic indexing) achieves the same definability patterns as structurally simplifying the frames (restricting to at most two flows), regardless of which temporal interpretation (standard or strict) is adopted for $G,H$. This equivalence demonstrates that indices essentially serve to isolate individual functions within a multiflow setting; once this structural complexity is eliminated, the basic temporal-modal operators of $L_{T\times W}$ already provide the same definability patterns for the functional properties considered.

In fact, the connection between minimal frames and indexed languages is not limited to the specific properties examined here. As shown in the Appendix (Theorem~\ref{thm:minimal-indexed}), it constitutes a general definability equivalence theorem: for any order type and any functional property that concerns a single function, definability in minimal frames is equivalent to definability in indexed frames. The full technical development---including the formal definition of minimal projection and the proof of the equivalence---is presented there.

\subsubsection{Comparison with Uniform-Domain Constraints}

Unlike minimal frames, U-Dom frames retain a genuine arbitrary ordered multiflow structure---the modal operators $\square$ and $\lozenge$ still quantify over all accessibility functions. The uniformity of domains restores the arrow-style characterisations in linear orders (strict or non-strict) that fail in arbitrary ordered multiflow frames (see Example~\ref{fallo}), because the aggregate image $\mathcal{F}_w(X)$ of a set $X \subseteq T_w$ now behaves as if it came from a single function. This uniformity is captured formally by the following condition: for every $w \in W$, $\operatorname{Dom}(f_{ww'}) = \operatorname{Dom}(f_{ww''})$ for all $f_{ww'}, f_{ww''} \in \mathcal{F}_w$; this common domain is denoted by $\operatorname{Dom}_U(\mathcal{F}_w)$.
\begin{theorem}[Arrow-style characterisations under U-Dom]\label{thm:UDom-arrow-characterizations}
Let $\Sigma = (W, \mathcal{T}, \mathcal{F})$ be a U-Dom frame. Then:

{\footnotesize
\begin{center}
\begin{tabular}{|l|c|l|}
\hline
\multicolumn{3}{|c|}{\emph{Linear orders (strict or non-strict)}} \\
\hline
\textbf{Property} & \textbf{Domain} & \textbf{Condition} \\
\hline
Totality          & $t_w \in \mathrm{Coord}_{\Sigma}$ 
                  & $\mathcal{F}_w(t_w\!\downarrow^{\ast}) \cup\, \mathcal{F}_w(t_w\!\uparrow^{\ast}) \subseteq 
                     \mathcal{F}_w(\{t_w\})\!\downarrow^{\ast} \cup\, \mathcal{F}_w(\{t_w\})\!\uparrow$ \\
\hline
Injectivity       & $t_w \in \mathrm{Dom}_U(\mathcal{F}_w)$ 
                  & $\mathcal{F}_w(t_w\!\downarrow^{\ast}) \cup\, \mathcal{F}_w(t_w\!\uparrow^{\ast}) \subseteq 
                     \mathcal{F}_w(\{t_w\})\!\downarrow^{\ast} \cup\, \mathcal{F}_w(\{t_w\})\!\uparrow^{\ast}$ \\
\hline
Surjectivity      & $t_w \in \mathrm{Coord}_{\Sigma}$ 
                  & $\mathcal{F}_w(\{t_w\})\!\downarrow^{\ast} \cup\, \mathcal{F}_w(\{t_w\})\!\uparrow^{\ast} \subseteq 
                     \mathcal{F}_w(t_w\!\downarrow^{\ast}) \cup\, \mathcal{F}_w(t_w\!\uparrow^{\ast})$ \\
\hline
Constancy         & $t_w \in \mathrm{Dom}_U(\mathcal{F}_w)$ 
                  & $\mathcal{F}_w(t_w\!\downarrow^{\ast}) \cup\, \mathcal{F}_w(t_w\!\uparrow^{\ast}) \subseteq 
                     \mathcal{F}_w(\{t_w\})$ \\
\hline
\noalign{\smallskip}
\hline
\multicolumn{3}{|c|}{\emph{All order types}} \\
\hline
\textbf{Property} & \textbf{Domain} & \textbf{Condition} \\
\hline
Monotonicity        & $t_w \in \mathrm{Dom}_U(\mathcal{F}_w)$ 
                  & $\mathcal{F}_w(t_w\!\uparrow^{\ast}) \subseteq \mathcal{F}_w(\{t_w\})\!\uparrow$ \\
\hline
Strictly mon.     & $t_w \in \mathrm{Dom}_U(\mathcal{F}_w)$ 
                  & $\mathcal{F}_w(t_w\!\uparrow^{\ast}) \subseteq \mathcal{F}_w(\{t_w\})\!\uparrow^{\ast}$ \\
\hline
Antitonicity        & $t_w \in \mathrm{Dom}_U(\mathcal{F}_w)$ 
                  & $\mathcal{F}_w(t_w\!\uparrow^{\ast}) \subseteq \mathcal{F}_w(\{t_w\})\!\downarrow$ \\
\hline
Strictly ant.     & $t_w \in \mathrm{Dom}_U(\mathcal{F}_w)$ 
                  & $\mathcal{F}_w(t_w\!\uparrow^{\ast}) \subseteq \mathcal{F}_w(\{t_w\})\!\downarrow^{\ast}$ \\
\hline
\end{tabular}
\end{center}
}
\end{theorem}

\medskip\noindent
These characterisations lead directly to defining formulas in $L_{T\times W}$. The formulas for totality, non-totality, and surjectivity from Section~\ref{sec:ordered-multiflow} remain valid under U-Dom and are omitted here. For the remaining properties, the U-Dom condition introduces the disjunct $\square\bot$ to handle points outside the common domain $\operatorname{Dom}_U(\mathcal{F}_w)$. We present the formulas separately for the two temporal interpretations. With the exception of totality and non-totality, which require specific treatment based on the order's reflexivity, these adapted formulas extend the results of \cite{Burrieza2003} to all order types.

\medskip\noindent
\textit{Under the standard interpretation ($G,H$).}
\begin{align*}
(UD\text{-}Inj)     &: \square(Hp \land Gp) \to (\square\bot \lor (H\square p \land G\square p)) \\[4pt]
(UD\text{-}Inc)    &: \square Gp \to (\square\bot \lor G\square p) \qquad\text{(non-strict orders)}\\[4pt]
(UD\text{-}StrInc)   &: \square(p \land Gp) \to (\square\bot \lor G\square p) \qquad\text{(strict orders)}\\[4pt]
(UD\text{-}Dec)    &: \square Gp \to (\square\bot \lor H\square p) \qquad\text{(non-strict orders)}\\[4pt]
(UD\text{-}StrDec)   &: \square(p \land Gp) \to (\square\bot \lor H\square p) \qquad\text{(strict orders)}\\[4pt]
(UD\text{-}Con)      &: \square p \to (\square\bot \lor (H\square p \land G\square p))
\end{align*}

\medskip\noindent
\textit{Under the strict interpretation ($G^\ast,H^\ast$).}
For the strict interpretation, the formulas from Lemma~\ref{lem:strict-persistence} carry over directly, with the addition of $\square\bot$ as above. Totality, non-totality, and surjectivity are given by the formulas in Lemma~\ref{lem:strict-persistence}(1a) with $G,H$ replaced by $G^\ast,H^\ast$; we omit them here. For the remaining properties, the strict interpretation already excludes the present, so the distinction between non-strict and strict orders disappears:
\begin{align*}
(UD\text{-}Inj)_\ast &: \square(H^\ast p \land G^\ast p) \to (\square\bot \lor (H^\ast\square p \land G^\ast\square p)) \\[4pt]
(UD\text{-}Inc)_\ast &: \square G^\ast p \to (\square\bot \lor G^\ast\square p) \\[4pt]
(UD\text{-}StrInc)_\ast &: \square(p \land G^\ast p) \to (\square\bot \lor G^\ast\square p) \\[4pt]
(UD\text{-}Dec)_\ast &: \square G^\ast p \to (\square\bot \lor H^\ast\square p) \\[4pt]
(UD\text{-}StrDec)_\ast &: \square(p \land G^\ast p) \to (\square\bot \lor H^\ast\square p) \\[4pt]
(UD\text{-}Con)_\ast &: \square p \to (\square\bot \lor (H^\ast\square p \land G^\ast\square p))
\end{align*}

\begin{theorem}[Coincidence with minimal frames]\label{thm:udom-coincidence}
For every order type $\Or$, a functional property $P$ is definable in minimal frames ($\Or^2$) if and only if it is definable in U-Dom frames of type $\Or$.  
In particular:
\begin{itemize}
    \item The ``hard core'' of properties---totality, non-totality, injectivity, surjectivity, constancy---remains undefinable in non-linear orders under both settings.
    \item Monotonicity, antitonicity and their strict variants are definable in all order types.
    \item Definability in linear orders is maximal: all nine properties considered become definable (with the appropriate interpretation, standard or strict).
\end{itemize}
\end{theorem}

\begin{proof}[Proof sketch]
The verification follows the same pattern established in Section~\ref{sec:ordered-multiflow} for totality and surjectivity: each arrow-style characterisation in Theorem~\ref{thm:UDom-arrow-characterizations} is translated into the corresponding modal formula via Lemma~\ref{lem:image-modal-unified}. The disjunct $\square\bot$ handles points outside the common domain $\operatorname{Dom}_U(\mathcal{F}_w)$, where no function is defined.

Undefinability follows from the same counter-examples used for minimal frames (see the Appendix). Since each involves at most two flows and one function, it trivially satisfies the U-Dom condition while witnessing the indefinability of the corresponding property.\hfill\qed
\end{proof}

\medskip\noindent
This coincidence reveals that functional multiplicity only obscures definability when functions are allowed to diverge in their domains. Once a common structural ground is established---whether through the geometric simplicity of minimal frames or the domain regularity of the U-Dom condition---the basic operators of $L_{T\times W}$ recover their full capacity to characterize functional mappings. Thus, while U-Dom frames retain a genuine multiflow architecture, their domain regularity yields exactly the same definability patterns for the functional properties under study as minimal frames.

\section{Conclusions and Future Work}
\label{sec:conclusions}

This work provides a systematic map of definability for basic functional properties in the modal-temporal language $L_{T\times W}$ across different order types. The results show that definability is governed by two factors: the locality of the Priorean operators (under both their standard ($G,H$) and strict ($G^*,H^*$) interpretations) and the structural complexity of the underlying frames.

In the original multiflow setting, definability is severely limited: only totality and surjectivity are definable, and only within linear orders. To overcome these limitations, we introduced \emph{minimal frames}, which restrict each frame to at most two flows and a single accessibility function. This simplification transforms the modal operators $\square, \lozenge$ from second-order quantifiers over arbitrary families of functions into first-order ones, yielding substantial gains in definability. In this setting, monotonicity and antitonicity become definable in all order types, and for linear orders the picture is nearly complete: all nine properties become definable, with the sole exception of injectivity in $\LO$ when evaluated under the standard interpretation.

Even in minimal frames, however, a core set of properties remains undefinable in non-linear orders. Totality, non-totality, injectivity, surjectivity, and constancy cannot be defined in $\PRE, \sPRE$, or $\PO$. This limitation stems from the inherent locality of temporal modalities, which cannot relate points in incomparable or isolated parts of the domain.

The \emph{strict interpretation} ($G^*,H^*$) partially overcomes this barrier. In minimal frames, it allows us to define strict monotonicity and strict antitonicity across all order types, and to recover injectivity in $\LO$. Nevertheless, totality, non-totality, surjectivity, and constancy remain undefinable in non-linear orders even under this strict reading.

A methodological point concerns constancy in the minimal setting. Unlike the other properties studied, constancy is preserved by surjective p-morphisms in minimal frames, making the usual counterexample technique inapplicable. We therefore used a semantic equivalence argument to show that, in non-linear orders, $L_{T\times W}$ cannot distinguish constant from non-constant functions. This exceptional behaviour highlights that constancy occupies a special place among the properties considered.

Our comparison with indexed languages and Uniform Domain (U-Dom) frames confirms that minimal frames capture the same definability patterns while maintaining a simpler architecture. Once the structural noise of multiple divergent domains is eliminated, the basic operators of $L_{T\times W}$ recover their full capacity to characterize functional mappings.
These results point to two natural directions for future research.
Both share a common trait: they address definability not by refining
the order type, but by introducing mechanisms that operate
independently of it.

The first is to strengthen the connectivity of temporal orders by
means of \emph{cohesive frames}, where every flow is finitely
connected \cite{HughesCresswell1984}.  Capturing the finite
bidirectional paths between any two points naturally leads to
infinitary connectives.  The key question is whether such an
enrichment can extend definability of the ``hard core'' beyond the
linear case, by shifting the focus from the order type to the
structural property of cohesiveness.

The second direction is to add operators that quantify over all points
in a flow independently of the temporal order, such as the universal
modality $[U]$ \cite{GorankoPassy1992} or the universal inequality
operator $[\neq]$ \cite{deRijke1992}.  These operators allow the logic
to ``jump'' across time.  Because they bypass the order entirely, they
are not appropriate for treating orders as such, but they offer a
radically different route to definability for arbitrary functional
frames.

\newpage

\appendix
\section*{Appendix}
\setcounter{section}{1}
\renewcommand{\thesection}{}
\renewcommand{\thesubsection}{\arabic{subsection}}

\renewcommand{\thetheorem}{A.\thesubsection.\arabic{theorem}}
\renewcommand{\thedefinition}{A.\thesubsection.\arabic{theorem}}
\renewcommand{\thelemma}{A.\thesubsection.\arabic{theorem}}
\renewcommand{\theproposition}{A.\thesubsection.\arabic{theorem}}
\renewcommand{\thecorollary}{A.\thesubsection.\arabic{theorem}}
\renewcommand{\thenotation}{A.\thesubsection.\arabic{theorem}}
\renewcommand{\theconvention}{A.\thesubsection.\arabic{theorem}}
\renewcommand{\theobservation}{A.\thesubsection.\arabic{theorem}}

\renewcommand{\theremark}{A.\thesubsection.\arabic{remark}}
\renewcommand{\theexample}{A.\thesubsection.\arabic{example}}

\setcounter{theorem}{0}
\setcounter{remark}{0}
\setcounter{example}{0}

\subsection{Summary of Defining Formulas}
Tables~\ref{tab:definability-formulas-general} and~\ref{tab:definability-formulas-strict} below
show that the strict interpretation ($G^\ast,H^\ast$) yields a more uniform pattern of definability 
across the minimal order types of the $O^2$ family. 
Several properties that require different schemata under the standard interpretation ($G,H$) admit a single, 
uniform formula under the strict reading.

\noindent This uniformity reflects a shift in the definability profile. The strict interpretation loses sensitivity to reflexive loops but gains uniformity for properties tied to irreflexive accessibility.

\bigskip
\begin{table}[!ht]
\centering
\scriptsize
\begin{tabular}{|l|l|l|}
\hline
\textbf{Property} & \textbf{Formula} & \textbf{Definable in} \\
\hline
Totality & $\square(Hp \land Gp)\to (H\square p \land G\square p)$ & $LO^{2}$ \\
         & $\square(Hp \land p \land Gp)\to (H\square p \land G\square p)$ & $sLO^{2}$ \\
\hline
Non-totality & $P\square\bot \lor F\square\bot$ & $LO^{2}$ \\
             & $P\square\bot \lor \square\bot \lor F\square\bot$ & $sLO^{2}$ \\
\hline
Injectivity & $\lozenge(Hp \land Gp)\to (H\square p \land G\square p)$ & $sLO^{2}$ \\
\hline
Surjectivity & $(H\square p \land G\square p)\to \square(Hp \land Gp)$ & $LO^{2},\ sLO^{2}$ \\
\hline
Monotonicity & $\lozenge Gp \to G\square p$ & $PRE^{2},\ PO^{2},\ LO^{2}$ \\
             & $\lozenge (p\land Gp) \to G\square p$ & $sPRE^{2},\ sLO^{2}$ \\
\hline
Strict Monotonicity & $\lozenge Gp\to G\square p$ & $sPRE^{2},\ sLO^{2}$ \\
\hline
Antitonicity & $\lozenge Gp \to H\square p$ & $PRE^{2},\ PO^{2},\ LO^{2}$ \\
             & $\lozenge (p\land Gp) \to H\square p$ & $sPRE^{2},\ sLO^{2}$ \\
\hline
Strict Antitonicity & $\lozenge Gp\to H\square p$ & $sPRE^{2},\ sLO^{2}$ \\
\hline
Constancy & $\lozenge p \to (H\square p \land G\square p)$ & $LO^{2},\ sLO^{2}$ \\
\hline
\end{tabular}
\caption{Defining formulas for functional properties under the standard interpretation ($G,H$).}
\label{tab:definability-formulas-general}
\end{table}

\begin{table}[!ht]
\centering
\scriptsize
\begin{tabularx}{\textwidth}{|l|X|l|}
\hline
\textbf{Property} & \textbf{Formula} & \textbf{Definable in} \\
\hline
Totality & $\square(H^\ast p \land p \land G^\ast p)\to (H^\ast\square p \land G^\ast\square p)$ & $LO^2,\ sLO^2$ \\
\hline
Non-totality & $P^\ast\square\bot \lor \square\bot \lor F^\ast\square\bot$ & $LO^2,\ sLO^2$ \\
\hline
Injectivity & $\lozenge(H^\ast p \land G^\ast p)\to (H^\ast\square p \land G^\ast\square p)$ & $LO^2,\ sLO^2$ \\
\hline
Surjectivity & $(H^\ast\square p \land G^\ast\square p)\to \square(H^\ast p \land G^\ast p)$ & $LO^2,\ sLO^2$ \\
\hline
Monotonicity & $\lozenge (p\land G^\ast p) \to G^\ast\square p$ & all $O^2$ \\
\hline
Strict Monotonicity & $\lozenge G^\ast p\to G^\ast\square p$ & all $O^2$ \\
\hline
Antitonicity & $\lozenge (p\land G^\ast p) \to H^\ast\square p$ & all $O^2$ \\
\hline
Strict Antitonicity & $\lozenge G^\ast p\to H^\ast\square p$ & all $O^2$ \\
\hline
Constancy & $\lozenge p \to (H^\ast\square p \land G^\ast\square p)$ & $LO^2,\ sLO^2$ \\
\hline
\end{tabularx}
\caption{Defining formulas for functional properties under the strict interpretation ($G^\ast,H^\ast$).}
\label{tab:definability-formulas-strict}
\end{table}

\begin{remark}
\noindent
\begin{enumerate}
    \item The label ``all $\Or^2$'' indicates definability across all minimal order types in the $\Or^2$ family.
    \item For properties of order preservation, mirror-image formulations using $H$ and $H^\ast$ are also possible and equivalent over minimal frames to the future-oriented ones listed.
\end{enumerate}
\end{remark}

\subsection{Formal correspondence between minimal and indexed frames}\label{comparacion}

We now establish a precise formal correspondence between the minimal frame family $\Or^2$ and the class of indexed frames of type $\Or$. An \emph{indexed frame} is a structure $\Sigma^{\mathfrak{I}} = (W,\mathcal{T},\mathcal{F})$ where $W$ is a set of worlds, $\mathcal{T}$ a family of orders $(T_w,R_w)$ as in Definition~\ref{definition1}, and $\mathcal{F}$ a family of accessibility functions $f_{wi}$ indexed by $i\in\mathfrak{I}$, each $f_{wi}: T_w \rightharpoonup T_i$. The modal operators are then $[i]$ and $\langle i\rangle$, interpreted as:
\[
h([i]A) = \{t_w \mid f_{wi}(\{t_w\})\subseteq h(A)\}, \qquad
h(\langle i\rangle A) = \{t_w \mid f_{wi}(\{t_w\})\cap h(A)\neq\varnothing\}.
\]
(We refer to \cite[Definitions 2.2--2.4]{Burrieza2009} for the full details.)

\begin{definition}[Minimal projection]
Let $\Sigma^{\mathfrak{I}} = (W,\mathcal{T},\mathcal{F})$ be an indexed frame of type $\Or$. For each index $i \in \mathfrak{I}$, we define the \emph{minimal projection} of $\Sigma^{\mathfrak{I}}$ with respect to $i$ as the frame
\[
\Sigma_{wi} = 
\begin{cases}
(\{w,i\}, \{(T_w,R_w),(T_i,R_i)\}, \{f_{wi}\}) & \text{if } i \in W \text{ and } f_{wi} \text{ exists},\\[4pt]
(\{w\}, \{(T_w,R_w)\}, \varnothing) & \text{otherwise}.
\end{cases}
\]
In all cases, $\Sigma_{wi}$ belongs to the minimal class $\Or^2$.
\end{definition}

\begin{lemma}[Validity transfer]\label{lem:validity-transfer}
Let $\Sigma^{\mathfrak{I}} = (W,\mathcal{T},\mathcal{F})$ be an indexed frame of type $\Or$ and let $A \in L_{T\times W}$. For each index $i \in \mathfrak{I}$, let $\Sigma_{wi}$ be the minimal projection defined above. Then:
\[
\Sigma^{\mathfrak{I}} \models \{A(i) \mid i \in \mathfrak{I}\} \quad\text{iff}\quad 
\Sigma_{wi} \models A \text{ for every } i \in \mathfrak{I}.
\]

\begin{proof}
The proof follows from the semantics of indexed languages \cite[Definition 2.4]{Burrieza2009}. For a fixed $i$, the formula $A(i)$ is obtained from $A$ by replacing $\square$ with $[i]$ and $\lozenge$ with $\langle i\rangle$. In the indexed frame, $[i]$ and $\langle i\rangle$ refer exclusively to the function $f_{wi}$ (if it exists) and to no other. A straightforward structural induction on $A$ shows that the truth value of $A(i)$ at any coordinate $t_w$ depends only on:
\begin{itemize}
    \item the interpretation of $[i]$ and $\langle i\rangle$, which involves only $f_{wi}$;
    \item the temporal operators $G$ and $H$, which involve only the orders on $T_w$ and $T_i$;
    \item the atomic formulas, whose valuation is inherited directly from $\Sigma^{\mathfrak{I}}$.
\end{itemize}
These are exactly the components preserved in the minimal projection $\Sigma_{wi}$. Hence $A(i)$ is valid in $\Sigma^{\mathfrak{I}}$ iff $A$ is valid in $\Sigma_{wi}$. The equivalence for the whole set $\{A(i)\}$ follows by quantifying over all $i$. \end{proof}
\end{lemma}

\begin{theorem}[Definability equivalence]\label{thm:minimal-indexed}
A functional property $P$ is $L^{\mathfrak I}$-definable in the class of indexed frames of type $\Or$ if and only if it is definable in the class of minimal frames $\Or^2$.
\end{theorem}

\begin{proof}
$(\Rightarrow)$ Suppose $P$ is $L_{T\times W}$-definable in minimal frames $\Or^2$ by a formula $A \in L_{T\times W}$. Let $\Sigma^{\mathfrak{I}}$ be an indexed frame of type $\Or$.

If every function $f_{wi}$ in $\Sigma^{\mathfrak{I}}$ satisfies $P$, then each minimal projection $\Sigma_{wi}$ validates $A$ (since $A$ defines $P$ in $\Or^2$). By Lemma~\ref{lem:validity-transfer}, the set $\{A(i) \mid i \in \mathfrak{I}\}$ is valid in $\Sigma^{\mathfrak{I}}$. 

If some function $f_{wi}$ in $\Sigma^{\mathfrak I}$ fails $P$, then its corresponding projection $\Sigma_{wi}$ falsifies $A$, so by the lemma the set $\{A(i)\}$ is not valid in $\Sigma^{\mathfrak{I}}$. 

Therefore $\{A(i)\}$ defines the class of indexed frames whose functions all satisfy $P$, which is exactly the $L^I$-definability of $P$ in $\Or$.

$(\Leftarrow)$ Let $\Gamma = \{A(i) \mid i \in \mathfrak{I}\}$ define $P$ in indexed frames of type $\Or$, i.e., $P$ is $L^I$-definable in $\Or$. Define $A$ by replacing every occurrence of $[i]$ with $\square$ and $\langle i\rangle$ with $\lozenge$ in any formula of $\Gamma$. (The choice of index is irrelevant, as all $A(i)$ are syntactically identical up to the index label.)

Take an arbitrary minimal frame $\Sigma_{ww'}$ of type $\Or^2$. Choose an index $i \in \mathfrak{I}$ and rename the flow $w'$ as $i$, obtaining an indexed frame $\Sigma^{\mathfrak{I}}$ with a single relevant index and the same temporal structure. By construction, $\Sigma_{ww'}$ satisfies $P$ iff $\Sigma^{\mathfrak{I}}$ satisfies $P$. Since $\Gamma$ defines $P$ in indexed frames of type $\Or$, we have that $\Gamma$ is valid in $\Sigma^{\mathfrak{I}}$ iff $\Sigma^{\mathfrak{I}}$ satisfies $P$. Applying Lemma~\ref{lem:validity-transfer}, we obtain that $\Gamma$ is valid in $\Sigma^{\mathfrak{I}}$ iff $A$ is valid in $\Sigma_{ww'}$. Therefore, $A$ is valid in $\Sigma_{ww'}$ iff $\Sigma_{ww'}$ satisfies $P$, i.e., $A$ defines $P$ in $\Or^2$. Hence $P$ is $L_{T\times W}$-definable in $\Or^2$.
\end{proof}

\begin{corollary}
The definability patterns established in Section~\ref{sec:minimal} for minimal frames (Tables~\ref{tab:definability-minimal-frames} and~\ref{tab:definability-minimal-frames-strict}) hold verbatim for indexed frames over the corresponding order types.
\end{corollary}

\subsection{Figures of undefinability}\label{appendix:figures}
The following eight figures provide a catalogue of counterexamples illustrating the indefinability of functional properties across different order types. Each figure shows a pair of frames connected by a surjective p-morphism, where the property in question holds in the source but fails in the target. The appendix thus complements the proofs in the text by offering a systematic gallery of undefinability patterns.

\medskip
\noindent\textsc{Conventions for establishing undefinability.}

\smallskip
\noindent\emph{Convention (RC).} All counterexample figures in the Appendix (except Figures~\ref{fig:inj_P-PO-LO} and~\ref{fig:strict-antitonicity_P-PO-LO}) depict strict orders. To obtain counterexamples for the reflexive classes $\PRE$, $\PO$, and $\LO$, we take the reflexive closure of the depicted strict order (i.e., we add a reflexive loop at every point). The resulting frame is a counterexample for the corresponding reflexive order type.

\medskip
\noindent Graphical elements (rectangles, circles, arrows) are explained only at their first occurrence; the explanation is not repeated in later figures. If a figure introduces no new elements, its caption will state: ``No new graphical elements are introduced in this figure.'' Figures that depict strict orders can be reused for reflexive classes by applying the reflexive closure (Convention~RC).

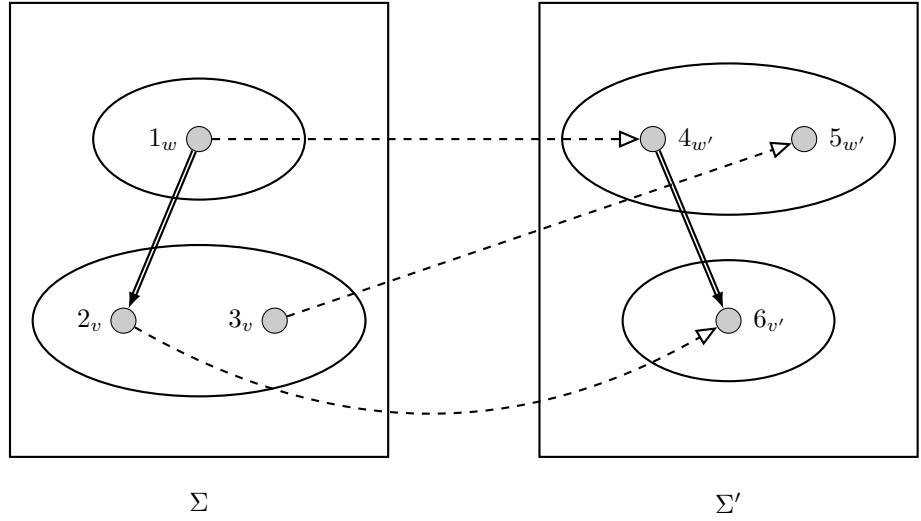
\begin{figure}[ht]
\begin{tikzpicture}[>=latex]
\usetikzlibrary{arrows.meta}
  \tikzset{
    hom/.style={->, thick, dashed, >=open triangle 45}
  }

  % Nodos
  \node[circle, draw=black, fill=black!20, minimum size=0.3cm] (A) at (-3.5,1.2) {}; % 1w
  \node[circle, draw=black, fill=black!20, minimum size=0.3cm] (C) at (-4.5,-1.2) {}; % 2v
  \node[circle, draw=black, fill=black!20, minimum size=0.3cm] (B) at (-2.5,-1.2) {}; % 3v
  \node[circle, draw=black, fill=black!20, minimum size=0.3cm] (D) at (2.5,1.2) {};  % 4w'
  \node[circle, draw=black, fill=black!20, minimum size=0.3cm] (E) at (3.5,-1.2) {}; % 6v'
  \node[circle, draw=black, fill=black!20, minimum size=0.3cm] (F) at (4.5,1.2) {}; % 5w'

  % Etiquetas
  \node[left=9pt] at (A.east) {$1_w$};
  \node[left=9pt] at (C.east) {$2_v$};
  \node[left=9pt] at (B.east) {$3_v$};
  \node[right=1pt] at (D.east) {$4_{w'}$};
  \node[right=1pt] at (F.east) {$5_{w'}$};
  \node[right=1pt] at (E.east) {$6_{v'}$};

  % Flechas funcionales (dobles sólidas)
  \draw[->, thick, double] (D) -- (E);     
  \draw[->, thick, double] (A) -- (C);     

  % Homomorfismo sobreyectivo (flechas estilizadas)
  \draw[hom] (B) -- (F);                
  \draw[hom] (A) -- (D);                
  \draw[hom, bend right] (C) to (E);    

  % Marcos y elipses
  \draw[thick] (-6,-3) rectangle (-1,3);               
  \draw[thick] (-3.5,1.2) ellipse (1.4 and 0.8);        
  \draw[thick] (-3.5,-1.2) ellipse (2.2 and 1.0);       

  \draw[thick] (1,-3) rectangle (6,3);                 
  \draw[thick] (3.5,1.2) ellipse (2.2 and 1.0);         
  \draw[thick] (3.5,-1.2) ellipse (1.4 and 0.8);        

  % Etiquetas Σ y Σ'
  \node[below=10pt] at (-3.5,-3) {$\Sigma$};
  \node[below=10pt] at (3.5,-3) {$\Sigma'$};

\end{tikzpicture}

\begin{center}
\begin{minipage}{0.9\textwidth}
\small
\textit{Each square represents a functional frame: the left frame satisfies the property under consideration (here, totality), while the right frame does not. Circular nodes with subscripts denote coordinates (elements of the flows, represented by ellipses). Double solid arrows indicate functional edges; dashed arrows with open triangle heads indicate the surjective p-morphism $Z$. We describe only the graphical elements appearing in the current figure; relations $R$ and functions $f$ follow the conventions stated above (including Convention~RC). This practice will be maintained throughout. 
The surjective p-morphism is: $Z=\{(1_{w},\,4_{w'}),\ (2_{v},\,6_{v'}),\ (3_{v},\,5_{w'})\}$.}
\end{minipage}
\end{center}

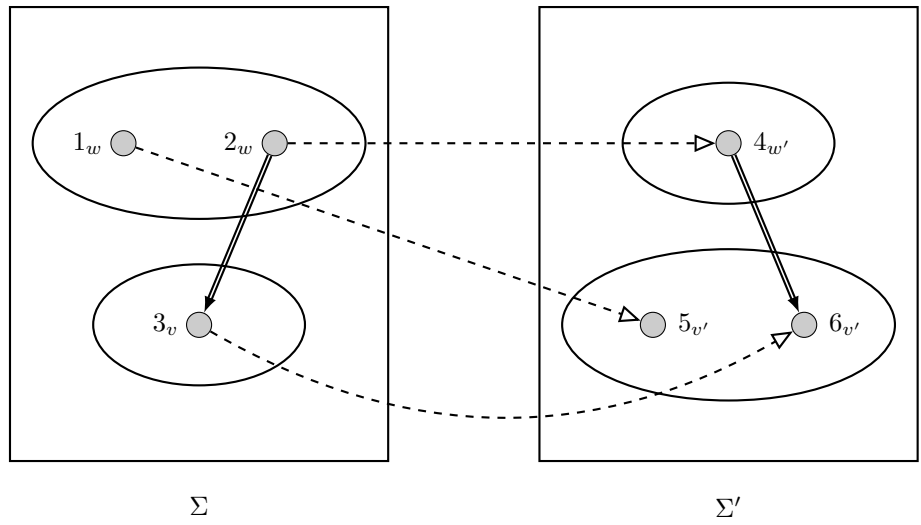
\caption{Undefinability of totality in strict preorders.
In $\Sigma$ totality holds, whereas in $\Sigma'$ it fails at $5_{w'}$, which has no image under $f_{w'v'}$.}

\label{fig:total}
\end{figure}

%%%
\begin{figure}[ht]
\begin{tikzpicture}[>=latex]

  % Nodos
  \node[circle, draw=black, fill=black!20, minimum size=0.3cm] (A) at (-4.5,1.2) {}; % 1w
  \node[circle, draw=black, fill=black!20, minimum size=0.3cm] (B) at (-2.5,1.2) {}; % 2w
  \node[circle, draw=black, fill=black!20, minimum size=0.3cm] (C) at (-3.5,-1.2) {}; % 3v
  \node[circle, draw=black, fill=black!20, minimum size=0.3cm] (D) at (3.5,1.2) {};  % 4w'
  \node[circle, draw=black, fill=black!20, minimum size=0.3cm] (E) at (2.5,-1.2) {}; % 5v′
  \node[circle, draw=black, fill=black!20, minimum size=0.3cm] (F) at (4.5,-1.2) {}; % 6v′
  
  % Etiquetas
  \node[left=9pt] at (A.east) {$1_w$};
  \node[left=9pt] at (B.east) {$2_w$};
  \node[left=9pt] at (C.east) {$3_v$};
  \node[right=1pt] at (D.east) {$4_{w'}$};
  \node[right=1pt] at (E.east) {$5_{v'}$};
  \node[right=1pt] at (F.east) {$6_{v'}$};
  
  % Flechas funcionales (dobles sólidas)
  \draw[->, thick, double] (B) to (C);                
  \draw[->, thick, double] (D) to (F);     

  % Homomorfismo sobreyectivo (discontinuas reforzadas) -- puntas Z igual que referencia
  \draw[->, thick, dashed, >=open triangle 45] (B) -- (D);                
  \draw[->, thick, dashed, bend right, >=open triangle 45] (C) to (F);                
  \draw[->, thick, dashed, >=open triangle 45] (A) -- (E);                

  % Marcos y elipses
  \draw[thick] (-6,-3) rectangle (-1,3);               
  \draw[thick] (-3.5,1.2) ellipse (2.2 and 1.0);        
  \draw[thick] (-3.5,-1.2) ellipse (1.4 and 0.8);       

  \draw[thick] (1,-3) rectangle (6,3);                 
  \draw[thick] (3.5,1.2) ellipse (1.4 and 0.8);         
  \draw[thick] (3.5,-1.2) ellipse (2.2 and 1.0);        

  % Etiquetas Σ y Σ'
  \node[below=10pt] at (-3.5,-3) {$\Sigma$};
  \node[below=10pt] at (3.5,-3) {$\Sigma'$};

\end{tikzpicture}

\begin{center}
\begin{minipage}{0.9\textwidth}
\small
\textit{Graphical conventions unchanged. \\[4pt] 
The surjective p-morphism is: $Z=\{(1_{w},\,5_{v'}),\ (2_{w},\,4_{w'}),\ (3_{v},\,6_{v'})\}$.}
\end{minipage}
\end{center}

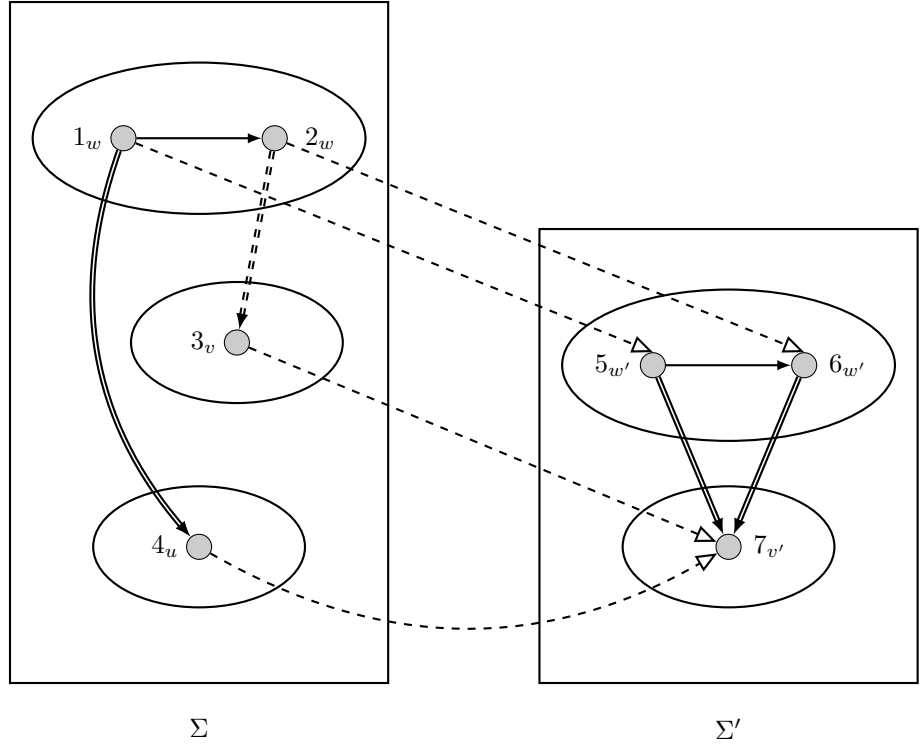
\caption{Undefinability of non-totality and surjectivity in strict preorders. In \(\Sigma\) the two properties hold, whereas in \(\Sigma'\) non-totality fails (since totality holds) and surjectivity fails because the codomain element \(5_{v'}\) has no preimage under \(f_{w'v'}\).}

\label{fig:surj}
\end{figure}

\par\bigskip
%%%%%%
\begin{figure}[ht]
\begin{tikzpicture}[>=latex]

  % Nodos
  \node[circle, draw=black, fill=black!20, minimum size=0.3cm] (A) at (-4.5,4.2) {}; % 1w
  \node[circle, draw=black, fill=black!20, minimum size=0.3cm] (B) at (-2.5,4.2) {}; % 2w
  \node[circle, draw=black, fill=black!20, minimum size=0.3cm] (C) at (-3.0,1.5) {}; % 3v
  \node[circle, draw=black, fill=black!20, minimum size=0.3cm] (D) at (4.5,1.2) {};  % 6w'
  \node[circle, draw=black, fill=black!20, minimum size=0.3cm] (E) at (3.5,-1.2) {}; % 7v′
  \node[circle, draw=black, fill=black!20, minimum size=0.3cm] (F) at (-3.5,-1.2) {}; % 4u
  \node[circle, draw=black, fill=black!20, minimum size=0.3cm] (G) at (2.5,1.2) {};  % 5w'
  
  % Etiquetas
  \node[left=9pt] at (A.east) {$1_w$};
  \node[right=3pt] at (B.east) {$2_w$};
  \node[left=9pt] at (C.east) {$3_v$};
  \node[right=1pt] at (D.east) {$6_{w'}$};
  \node[right=1pt] at (E.east) {$7_{v'}$};
  \node[left=9pt] at (F.east) {$4_u$};
  \node[left=9pt] at (G.east) {$5_{w'}$};
  
  % Flechas internas
  \draw[->, thick] (A) -- (B);                  
  \draw[->, thick, double, dashed] (B) -- (C);  
  \draw[->, thick, double] (D) -- (E);          
  \draw[->, thick, double, bend right] (A) to (F); 
  \draw[->, thick] (G) -- (D);                  
  \draw[->, thick, double] (G) -- (E);          
  
  % Homomorfismo redirigido: flechas Z con la misma punta abierta de referencia
  \draw[->, thick, dashed, >=open triangle 45] (B) -- (D.north);          
  \draw[->, thick, dashed, >=open triangle 45] (C) -- (E);          
  \draw[->, thick, dashed, bend right=30, >=open triangle 45] (F) to (E); 
  \draw[->, thick, dashed, >=open triangle 45] (A) -- (G.north);          
  
  % Marcos y elipses
  \draw[thick] (-6,-3) rectangle (-1,6);               
  \draw[thick] (-3.5,4.2) ellipse (2.2 and 1.0);    
  \draw[thick] (-3.0,1.5) ellipse (1.4 and 0.8);        
  \draw[thick] (-3.5,-1.2) ellipse (1.4 and 0.8);       

  \draw[thick] (1,-3) rectangle (6,3);                 
  \draw[thick] (3.5,1.2) ellipse (2.2 and 1.0);         
  \draw[thick] (3.5,-1.2) ellipse (1.4 and 0.8);        

  % Etiquetas Σ y Σ′
  \node[below=10pt] at (-3.5,-3) {$\Sigma$};
  \node[below=10pt] at (3.5,-3) {$\Sigma'$};

\end{tikzpicture}

\begin{center}
\begin{minipage}{0.9\textwidth}
\small
\textit{This is the first example where the left frame contains three flows. 
Single solid arrows indicate precedence between elements. 
Dashed double arrows, like double solid arrows, denote functional relations; the two styles distinguish different functions originating from the same domain. 
Curved functional arrows are also introduced in this figure. \\[4pt]
The surjective p-morphim is: $Z=\{(1_{w},\,5_{w'}),\ (2_{w},\,6_{w'}),\ (3_{v},\,7_{v'}),\ (4_{u},\,7_{v'})\}$.}

\end{minipage}
\end{center}
\caption{Undefinability of non-totality, injectivity, and strict monotonicity in strict preorders and strict linear orders. In \(\Sigma\) the three properties hold, whereas in \(\Sigma'\) non-totality fails (totality holds instead), injectivity fails because distinct elements share the same image, and strict monotonicity fails for the same reason.}

\label{fig:inj_SPO-SOL}
\end{figure}
%%%%

\begin{figure}[ht]
\begin{tikzpicture}[>=latex]

  % Nodos
  \node[circle, draw=black, fill=black!20, minimum size=0.3cm] (A) at (-4.5,4.2) {}; % 1w
  \node[circle, draw=black, fill=black!20, minimum size=0.3cm] (B) at (-2.5,4.2) {}; % 2w
  \node[circle, draw=black, fill=black!20, minimum size=0.3cm] (C) at (-2.5,1.5) {}; % 4v
  \node[circle, draw=black, fill=black!20, minimum size=0.3cm] (D) at (4.5,1.2) {};  % 8w'
  \node[circle, draw=black, fill=black!20, minimum size=0.3cm] (E) at (4.5,-1.2) {}; % 10v′
  \node[circle, draw=black, fill=black!20, minimum size=0.3cm] (F) at (-4.5,-1.2){}; % 5u
  \node[circle, draw=black, fill=black!20, minimum size=0.3cm] (G) at (2.5,1.2){};   % 7w'
  \node[circle, draw=black, fill=black!20, minimum size=0.3cm] (H) at (-4.5,1.5){};  % 3v
  \node[circle, draw=black, fill=black!20, minimum size=0.3cm] (I) at (-2.5,-1.2){}; % 6u
  \node[circle, draw=black, fill=black!20, minimum size=0.3cm] (J) at (2.5,-1.2){};  % 9v'
    
  % Etiquetas
  \node[left=9pt] at (A.east) {$1_w$};
  \node[right=3pt] at (B.east) {$2_w$};
  \node[right=3pt] at (C.east) {$4_v$};
  \node[right=1pt] at (D.east) {$8_{w'}$};
  \node[right=1pt] at (E.east) {$10_{v'}$};
  \node[left=9pt] at (F.east) {$5_u$};
  \node[left=9pt] at (G.east) {$7_{w'}$};
  \node[left=9pt] at (H.east) {$3_v$};
  \node[right=3pt] at (I.east) {$6_u$};
  \node[left=12pt] at (J.east) {$9_{v'}$};
          
  % Flechas
  \draw[->, thick] (A) -- (B);                  
  \draw[->, thick, double] (D) -- (J);          

  % Flecha 2w → 8w' elevada al norte del nodo
  \draw[->, thick, dashed, >=open triangle 45] (B) -- (D.north);                

  % Z a 10v′ (ajustada homogénea)
  \draw[->, thick, dashed, shorten <=2pt, bend right, >=open triangle 45] 
    (I) to (E);   % 6u → 10v'
  \draw[->, thick, dashed, shorten <=2pt, >=open triangle 45] 
    (C) -- (E.135);   % 4v → 10v'

  % 5u → 9v'
  \draw[->, thick, dashed, bend right, >=open triangle 45] (F) to (J);    

  \draw[->, thick] (G) -- (D);                  
  % 1w → 7w'
  \draw[->, thick, dashed, >=open triangle 45] (A) to (G);                

  % 7w' → 10v' curvada para entrar por arriba casi vertical
  \draw[->, thick, double] (G) to[out=-20, in=90] (E);          

  % 3v → 9v'
  \draw[->, thick, dashed, >=open triangle 45] (H) -- (J);                

  \draw[->, thick] (H) -- (C);                  
  \draw[->, thick] (F) -- (I);                  
  \draw[->, thick] (J) -- (E);                  
  \draw[->, thick, double, dashed, >=latex] (B) -- (F);    
  \draw[->, thick, double] (A) -- (C);          

  % Marcos y elipses
  \draw[thick] (-6,-3) rectangle (-1,6);               
  \draw[thick] (-3.5,4.2) ellipse (2.2 and 1.0);    
  \draw[thick] (-3.5,1.5) ellipse (2.2 and 1.0);        
  \draw[thick] (-3.5,-1.2) ellipse (2.2 and 1.0);       

  \draw[thick] (1,-3) rectangle (6,3);                 
  \draw[thick] (3.5,1.2) ellipse (2.2 and 1.0);         
  \draw[thick] (3.5,-1.2) ellipse (2.2 and 1.0);        

  % Etiquetas Σ y Σ′
  \node[below=10pt] at (-3.5,-3) {$\Sigma$};
  \node[below=10pt] at (3.5,-3) {$\Sigma'$};

\end{tikzpicture}

\begin{center}
\begin{minipage}{0.9\textwidth}
\small
\textit{Graphical conventions unchanged; we list only the surjective p-morphism:\\
$Z=\{(1_{w},\,7_{w'}),\ (2_{w},\,8_{w'}),\ (3_{v},\,9_{v'}),\ (4_{v},\,10_{v'}),\ (5_{u},\,9_{v'}),\ (6_{u},\,10_{v'})\}$.}

\end{minipage}
\end{center}

\caption{Undefinability of monotonicity and constancy in strict preorders and strict linear orders.
In $\Sigma$ both properties hold, whereas in $\Sigma'$ monotonicity fails (the strict order is reversed) and constancy fails (the images are not all equal).}

\label{fig:inc}
\end{figure}

%%%
\begin{figure}[ht]
\begin{tikzpicture}[>=latex]

  % Nodos
  \node[circle, draw=black, fill=black!20, minimum size=0.3cm] (A) at (-4.5,4.2) {}; % 1w
  \node[circle, draw=black, fill=black!20, minimum size=0.3cm] (B) at (-2.5,4.2) {}; % 2w
  \node[circle, draw=black, fill=black!20, minimum size=0.3cm] (C) at (-2.5,1.5) {}; % 4v
  \node[circle, draw=black, fill=black!20, minimum size=0.3cm] (D) at (4.5,1.2) {};  % 8w'
  \node[circle, draw=black, fill=black!20, minimum size=0.3cm] (E) at (4.5,-1.2) {}; % 10v′
  \node[circle, draw=black, fill=black!20, minimum size=0.3cm] (F) at (-4.5,-1.2){}; % 5u
  \node[circle, draw=black, fill=black!20, minimum size=0.3cm] (G) at (2.5,1.2){};   % 7w'
  \node[circle, draw=black, fill=black!20, minimum size=0.3cm] (H) at (-4.5,1.5){};  % 3v
  \node[circle, draw=black, fill=black!20, minimum size=0.3cm] (I) at (-2.5,-1.2){}; % 6u
  \node[circle, draw=black, fill=black!20, minimum size=0.3cm] (J) at (2.5,-1.2){};  % 9v'
    
  % Etiquetas
  \node[left=9pt] at (A.east) {$1_w$};
  \node[right=3pt] at (B.east) {$2_w$};
  \node[right=3pt] at (C.east) {$4_v$};
  \node[right=1pt] at (D.east) {$8_{w'}$};
  \node[right=1pt] at (E.east) {$10_{v'}$};
  \node[left=9pt] at (F.east) {$5_u$};
  \node[left=9pt] at (G.east) {$7_{w'}$};
  \node[left=9pt] at (H.east) {$3_{v}$};
  \node[right=3pt] at (I.east) {$6_{u}$};
  \node[left=12pt] at (J.east) {$9_{v'}$};
          
  % Flechas
  \draw[->, thick] (A) -- (B);                  
  \draw[->, thick, double] (D) -- (E);                

  % Flecha 2w → 8w' (Z)
  \draw[->, thick, dashed, >=open triangle 45] (B) -- (D.north);                

  % Z a 10v′ (ajustada homogénea)
  \draw[->, thick, dashed, shorten <=2pt, bend right, >=open triangle 45] 
    (I) to (E);   % 6u → 10v'
  \draw[->, thick, dashed, shorten <=2pt, >=open triangle 45] 
    (C) -- (E.135);   % 4v → 10v'

  % 2w → 6u (NO Z, punta tipo latex)
  \draw[->, thick, double, dashed, shorten <=2pt, bend right, >=latex] 
    (B) to (I);

  \draw[->, thick, shorten <=2pt] 
    (F) -- (I);                     % 5u → 6u sólida

  % 5u → 9v' (Z)
  \draw[->, thick, dashed, bend right, >=open triangle 45] (F) to (J);    
  \draw[->, thick] (G) -- (D);                  
  % 1w → 7w' (Z)
  \draw[->, thick, dashed, >=open triangle 45] (A) -- (G.north);  
  \draw[->, thick, double] (G) -- (J);          
  % 3v → 9v' (Z)
  \draw[->, thick, dashed, >=open triangle 45] (H) -- (J);                
  \draw[->, thick] (H) -- (C);                  
  \draw[->, thick] (J) -- (E);                  
  \draw[->, thick, double] (A) -- (H);          

  % Marcos y elipses
  \draw[thick] (-6,-3) rectangle (-1,6);               
  \draw[thick] (-3.5,4.2) ellipse (2.2 and 1.0);    
  \draw[thick] (-3.5,1.5) ellipse (2.2 and 1.0);        
  \draw[thick] (-3.5,-1.2) ellipse (2.2 and 1.0);       

  \draw[thick] (1,-3) rectangle (6,3);                 
  \draw[thick] (3.5,1.2) ellipse (2.2 and 1.0);         
  \draw[thick] (3.5,-1.2) ellipse (2.2 and 1.0);        

  % Etiquetas Σ y Σ′
  \node[below=10pt] at (-3.5,-3) {$\Sigma$};
  \node[below=10pt] at (3.5,-3) {$\Sigma'$};

\end{tikzpicture}

\begin{center}
\begin{minipage}{0.9\textwidth}
\small
\textit{Graphical conventions unchanged; we list only the surjective p-morphism:\\
$Z=\{(1_{w},\,7_{w'}),\ (2_{w},\,8_{w'}),\ (3_{v},\,9_{v'}),\ (4_{v},\,10_{v'}),\ (5_{u},\,9_{v'}),\ (6_{u},\,10_{v'})\}$.}

\end{minipage}
\end{center}
\caption{Undefinability of antitonicity and strict antitonicity in strict preorders and strict linear orders.
In $\Sigma$ both properties hold, while in $\Sigma'$ antitonicity fails (a strict order in the domain is preserved instead of reversed), and strict antitonicity fails for the same reason.}

\label{fig:dec}
\end{figure}
%%%

\begin{figure}[ht]
\begin{tikzpicture}[>=latex]

  % Nodos izquierda (Σ)
  \node[circle, draw=black, fill=black!20, minimum size=0.3cm] (G) at (-4.5,1.2) {};  %1_w
  \node[circle, draw=black, fill=black!20, minimum size=0.3cm] (A) at (-2.5,1.2) {}; %2_w
  \node[circle, draw=black, fill=black!20, minimum size=0.3cm] (C) at (-4.5,-1.2) {}; %3_w
  \node[circle, draw=black, fill=black!20, minimum size=0.3cm] (B) at (-2.5,-1.2) {}; %4_w
  
  % Nodos derecha (Σ')
  \node[circle, draw=black, fill=black!20, minimum size=0.3cm] (D) at (2.5,1.2) {};  %5_{w'}
  \node[circle, draw=black, fill=black!20, minimum size=0.3cm] (F) at (4.5,1.2) {}; %6_{w'}
  \node[circle, draw=black, fill=black!20, minimum size=0.3cm] (E) at (3.5,-1.2) {}; %7_{w'}

  % Etiquetas
  \node[left=9pt] at (G.east) {$1_w$};
  \node[right=1pt] at (A.east) {$2_w$};
  \node[left=9pt] at (C.east) {$3_w$};
  \node[right=1pt] at (B.east) {$4_w$};
  \node[left=9pt] at (D.east) {$5_{w'}$};
  \node[right=1pt] at (F.east) {$6_{w'}$};
  \node[right=1pt] at (E.east) {$7_{w'}$};

  % Flechas funcionales (dobles sólidas)
  \draw[->, thick, double] (G) to (C); % 1_w → 3_w
  \draw[->, thick, double] (A) to (B); % 2_w → 4_w
  \draw[->, thick, double] (F) to (E); % 6_{w'} → 7_{w'}
  \draw[->, thick, double] (D) to (E); % 5_{w'} → 7_{w'}

  % Homomorfismo (flechas discontinuas)
  \draw[->, thick, dashed, bend right] (B) to (E); % 4_w → 7_{w'}
  \draw[->, thick, dashed, bend left] (A) to (D); % 2_w → 5_{w'}
  \draw[->, thick, dashed, bend left] (G) to (F); % 1_w → 6_{w'}
  \draw[->, thick, dashed, bend right] (C) to (E); % 3_w → 7_{w'}

  % Elipses: un poco más anchas para que las etiquetas no rocen
  \draw[thick] (-3.5,0) ellipse (2.3 and 2.2);  % elipse en Σ
  \draw[thick] (3.5,0) ellipse (2.3 and 2.2);   % elipse en Σ'

  % Rectángulos
  \draw[thick] (-6,-3) rectangle (-1,3);
  \draw[thick] (1,-3) rectangle (6,3);

  % Etiquetas de los marcos
  \node[below=10pt] at (-3.5,-3) {$\Sigma$};
  \node[below=10pt] at (3.5,-3) {$\Sigma'$};

\end{tikzpicture}
\begin{center}
\begin{minipage}{0.9\textwidth}
\small
\textit{Graphical conventions unchanged; we list only the surjective p-morphism:\\
$Z=\{(1_{w},\,5_{w'}),\ (2_{w},\,6_{w'}),\ (3_w,\,7_{w'}),\ (4_w,\,7_{w'})\}$.}
\end{minipage}
\end{center}
\caption{Undefinability of injectivity in strict preorders. In $\Sigma$, the function $f_{ww}$ is injective; in $\Sigma'$, the function $f_{w'w'}$ is not injective (both $5_{w'}$ and $6_{w'}$ map to $7_{w'}$).}
\label{fig:inj_SPRE2}
\end{figure}

%%%

\begin{figure}[ht]
\centering
\begin{tikzpicture}[>=latex]

  % Nodos
  \node[circle, draw=black, fill=black!20, minimum size=0.3cm] (A) at (-2.5,1.2) {}; % 2w 
  \node[circle, draw=black, fill=black!20, minimum size=0.3cm] (C) at (-4.5,-1.2) {}; % 3v
  \node[circle, draw=black, fill=black!20, minimum size=0.3cm] (B) at (-2.5,-1.2) {}; % 4v
  \node[circle, draw=black, fill=black!20, minimum size=0.3cm] (D) at (2.5,1.2) {};  % 5w' 
  \node[circle, draw=black, fill=black!20, minimum size=0.3cm] (E) at (3.5,-1.2) {}; % 7v' 
  \node[circle, draw=black, fill=black!20, minimum size=0.3cm] (F) at (4.5,1.2) {}; % 6w'
  \node[circle, draw=black, fill=black!20, minimum size=0.3cm] (G) at (-4.5,1.2) {};  % 1w 

  % Flechas reflexivas
  \draw[->, thick] (A) edge[loop above] ();
  \draw[->, thick] (B) edge[loop below] ();
  \draw[->, thick] (C) edge[loop below] ();
  \draw[->, thick] (D) edge[loop above] ();
  \draw[->, thick] (E) edge[loop below] ();
  \draw[->, thick] (F) edge[loop above] ();
  \draw[->, thick] (G) edge[loop above] ();
    
  % Etiquetas
  \node[right=1pt] at (A.east) {$2_w$};
  \node[left=9pt]  at (C.east) {$3_v$};
  \node[right=1pt] at (B.east) {$4_v$};
  \node[left=9pt]  at (D.east) {$5_{w'}$};
  \node[right=1pt] at (F.east) {$6_{w'}$};
  \node[right=1pt] at (E.east) {$7_{v'}$};
  \node[left=9pt]  at (G.east) {$1_w$};

  % Flechas funcionales y de orden
  \draw[->, thick, double] (F) -- (E); % 6w' → 7v'
  \draw[->, thick, double] (D) -- (E); % 5w' → 7v'
  \draw[->, thick, double] (G) -- (C); % 1w → 3v
  \draw[->, thick, double] (A) -- (B); % 2w → 4v
  \draw[->, thick]        (G) -- (A); % 1w → 2w
  \draw[->, thick]        (C) -- (B); % 3v → 4v
  \draw[->, thick]        (D) -- (F); % 5w' → 6w'

  % Flechas de Z con punta triangular hueca (ajuste milimétrico: shorten >=0.2pt)
  % 4v → 7v'
  \draw[->, thick, dashed, bend right, >=open triangle 45, shorten >=0.2pt] (B) to (E);
  % 2w → 6w'
  \draw[->, thick, dashed, bend right, >=open triangle 45, shorten >=0.2pt] (A) to (F);
  % 1w → 5w'
  \draw[->, thick, dashed, bend left,  >=open triangle 45, shorten >=0.2pt] (G) to (D);
  % 3v → 7v'
  \draw[->, thick, dashed, bend right, >=open triangle 45, shorten >=0.2pt] (C) to (E);

  % Marcos y elipses
  \draw[thick] (-6,-3) rectangle (-1,3);               
  \draw[thick] (-3.5,1.2) ellipse (2.2 and 1.0);        
  \draw[thick] (-3.5,-1.2) ellipse (2.2 and 1.0);       

  \draw[thick] (1,-3) rectangle (6,3);                 
  \draw[thick] (3.5,1.2) ellipse (2.2 and 1.0);         
  \draw[thick] (3.5,-1.2) ellipse (1.4 and 0.8);        

  % Etiquetas Σ y Σ′
  \node[below=10pt] at (-3.5,-3) {$\Sigma$};
  \node[below=10pt] at (3.5,-3) {$\Sigma'$};

\end{tikzpicture}

\begin{center}
\begin{minipage}{0.9\textwidth}
\small
\textit{Graphical conventions unchanged; we list only the surjective p-morphism:\\
$Z=\{(1_{w},\,5_{w'}),\ (2_{w},\,6_{w'}),\ (3_{v},\,7_{v'}),\ (4_{v},\,7_{v'})\}$.}
\end{minipage}
\end{center}

\caption{Undefinability of injectivity and strict monotonicity in preorders, posets, and linear orders. In \(\Sigma\) both properties hold, whereas in \(\Sigma'\) injectivity fails because distinct elements are mapped to the same image, and strict monotonicity fails for the same reason.}

\label{fig:inj_P-PO-LO}
\end{figure}

%%%
\begin{figure}[ht]
\begin{tikzpicture}[>=latex]

  % Nodos
  \node[circle, draw=black, fill=black!20, minimum size=0.3cm] (A) at (-2.5,1.2) {}; % 2w 
  \node[circle, draw=black, fill=black!20, minimum size=0.3cm] (C) at (-4.5,-1.2) {}; % 3v
  \node[circle, draw=black, fill=black!20, minimum size=0.3cm] (B) at (-2.5,-1.2) {}; % 4v
  \node[circle, draw=black, fill=black!20, minimum size=0.3cm] (D) at (2.5,1.2) {};  % 5w' 
  \node[circle, draw=black, fill=black!20, minimum size=0.3cm] (E) at (3.5,-1.2) {}; % 7v' 
  \node[circle, draw=black, fill=black!20, minimum size=0.3cm] (F) at (4.5,1.2) {}; % 6w'
  \node[circle, draw=black, fill=black!20, minimum size=0.3cm] (G) at (-4.5,1.2) {};  % 1w 

  % Flechas reflexivas
  \draw[->, thick] (A) edge[loop above] ();
  \draw[->, thick] (B) edge[loop below] ();
  \draw[->, thick] (C) edge[loop below] ();
  \draw[->, thick] (D) edge[loop above] ();
  \draw[->, thick] (E) edge[loop below] ();
  \draw[->, thick] (F) edge[loop above] ();
  \draw[->, thick] (G) edge[loop above] ();
    
  % Etiquetas
  \node[right=1pt] at (A.east) {$2_w$};
  \node[left=9pt]  at (C.east) {$3_v$};
  \node[right=1pt] at (B.east) {$4_v$};
  \node[left=9pt]  at (D.east) {$5_{w'}$};
  \node[right=1pt] at (F.east) {$6_{w'}$};
  \node[right=1pt] at (E.east) {$7_{v'}$};
  \node[left=9pt]  at (G.east) {$1_w$};

  % Flechas funcionales y de orden
  \draw[->, thick, double] (F) -- (E); % 6w' → 7v'
  \draw[->, thick, double] (D) -- (E); % 5w' → 7v'
  \draw[->, thick, double] (G) -- (B); % 1w → 4v
  \draw[->, thick, double] (A) -- (C); % 2w → 3v
  \draw[->, thick]        (G) -- (A); % 1w → 2w
  \draw[->, thick]        (C) -- (B); % 3v → 4v
  \draw[->, thick]        (D) -- (F); % 5w' → 6w'

  % Flechas de Z con punta open triangle 45
  % 4v → 7v'
  \draw[->, thick, dashed, bend right, >=open triangle 45] (B) to (E);
  % 2w → 6w' (ajustada con bend right)
  \draw[->, thick, dashed, bend right, >=open triangle 45] (A) to (F);
  % 1w → 5w'
  \draw[->, thick, dashed, bend left, >=open triangle 45] (G) to (D);
  % 3v → 7v'
  \draw[->, thick, dashed, bend right, >=open triangle 45] (C) to (E);

  % Marcos y elipses
  \draw[thick] (-6,-3) rectangle (-1,3);               
  \draw[thick] (-3.5,1.2) ellipse (2.2 and 1.0);        
  \draw[thick] (-3.5,-1.2) ellipse (2.2 and 1.0);       

  \draw[thick] (1,-3) rectangle (6,3);                 
  \draw[thick] (3.5,1.2) ellipse (2.2 and 1.0);         
  \draw[thick] (3.5,-1.2) ellipse (1.4 and 0.8);        

  % Etiquetas Σ y Σ′
  \node[below=10pt] at (-3.5,-3) {$\Sigma$};
  \node[below=10pt] at (3.5,-3) {$\Sigma'$};

\end{tikzpicture}

\begin{center}
\begin{minipage}{0.9\textwidth}
\small
\textit{Graphical conventions unchanged; we list only the surjective p-morphism:\\
$Z=\{(1_{w},\,5_{w'}),\ (2_{w},\,6_{w'}),\ (3_{v},\,7_{v'}),\ (4_{v},\,7_{v'})\}$.}

\end{minipage}
\end{center}
\caption{Undefinability of strict antitonicity in preorders, posets, and linear orders.
In $\Sigma$ strict antitonicity holds, whereas in $\Sigma'$ it fails (the strict order is collapsed to equality rather than reversed).}

\label{fig:strict-antitonicity_P-PO-LO}
\end{figure}
%%%

\end{document}